\documentclass[10pt,a4paper,titling,fleqn]{article} 

\usepackage[paper=a4paper,left=25mm,right=25mm,top=25mm,bottom=25mm]{geometry}

\usepackage{amsthm,amsmath,amsfonts,amssymb}
\usepackage{graphicx,psfrag,epsf}
\usepackage{enumerate}
\usepackage[numbers]{natbib}
\usepackage{url}
\usepackage{hyperref} 
\hypersetup{colorlinks=true,allcolors=[rgb]{0,0,1}}
\usepackage{color,xcolor}
\usepackage{dsfont}
\usepackage{multirow}
\usepackage{array}
\usepackage{booktabs}
\emergencystretch=\maxdimen
\allowdisplaybreaks

\selectfont
\newtheorem{theorem}{Theorem}[section] 
\newtheorem{definition}[theorem]{Definition}
\newtheorem{lemma}[theorem]{Lemma}
\newtheorem{corollary}[theorem]{Corollary}
\newtheorem{remark}[theorem]{Remark}
\newtheorem{example}[theorem]{Example}

\DeclareSymbolFont{symbolsC}{U}{txsyc}{m}{n}
\DeclareMathSymbol{\Nearrow}{\mathrel}{symbolsC}{116}
\DeclareMathSymbol{\Searrow}{\mathrel}{symbolsC}{117}
\DeclareMathSymbol{\Swarrow}{\mathrel}{symbolsC}{119}

\counterwithin{figure}{section}
\counterwithin{table}{section}
\counterwithin{equation}{section} 

\newcommand{\bx}{\mathbf{x}}
\newcommand{\by}{\mathbf{y}}
\newcommand{\bz}{\mathbf{z}}
\newcommand{\bu}{\mathbf{u}}
\newcommand{\bv}{\mathbf{v}}

\newcommand{\bU}{\mathbf{U}}

\newcommand{\bX}{\mathbf{X}}
\newcommand{\bY}{\mathbf{Y}}
\newcommand{\bZ}{\mathbf{Z}}

\newcommand{\bbP}{\mathbb{P}}
\newcommand{\bbR}{\mathbb{R}}
\newcommand{\bbI}{[0,1]}

\newcommand{\bZERO}{\boldsymbol{0}}

\newcommand{\cN}{\mathcal{N}}

\newcommand{\de}{\mathrm{\,d}}

\newcommand{\eqd}{\stackrel{\mathrm{d}}=}

\newcommand{\PLOD}{{\rm PLOD}}
\newcommand{\LTD}{{\rm LTD}}
\newcommand{\RTI}{{\rm RTI}}
\newcommand{\SI}{{\rm SI}}
\newcommand{\LOSD}{{\rm SI}}

\newcommand{\LOPDS}{{\rm PDS}}

\newcommand{\CIS}{\rm {CIS}}
\newcommand{\CI}{\rm {CI}}
\newcommand{\mtp}{{\rm mtp_2}}
\newcommand{\mrr}{{\rm mrr_2}}
\newcommand{\MTP}{{\rm MTP_2}}

\newcommand{\LOMTP}{{\rm MTP_2}}

\newcommand{\cMTP}{{\rm cMTP_2}}

\newcommand{\cLOMTP}{{\rm cMTP_2}}       
\newcommand{\cLOMTPs}{{\rm cMTP_2^{\rm set}}}

\newcommand{\dMTP}{{\rm dMTP_2}}

\newcommand{\cLOsMTP}{{\rm cMTP_2}}

\newcommand{\SF}[1]{{\color{black} #1}}

\definecolor{foxred}{rgb}{0.7, 0.11, 0.11}
\newcommand{\Mail}[1]{{\color{foxred} #1}}

\DeclareMathOperator*{\Motimes}{\text{\raisebox{0.25ex}{\scalebox{0.8}{$\bigotimes$}}}}

\begin{document}

\title{\bf \texorpdfstring{\SF{An MTP$_2$ property for conditional distributions}}}

\author{%
	Sebastian Fuchs\footnote{University of Salzburg, Austria, Email: \Mail{sebastian.fuchs@plus.ac.at}}
	\qquad
	Yuping Wang\footnote{University of Salzburg, Austria, Email: \Mail{yuping.wang@plus.ac.at}} 
    \vspace{5mm}
}

\maketitle

\begin{abstract}
We introduce a new notion of positive dependence, namely multivariate total positivity of order two for conditional distributions, denoted \(\cMTP\), that explicitly incorporates the conditioning variables.
This property is weaker than multivariate total positivity of order two for densities, but stronger than multivariate total positivity of order two for distribution functions, stochastic monotonicity, and tail monotonicity. 
It thus provides a natural intermediate concept linking these classical notions of positive dependence.
We verify the \(\cMTP\) property for several distributions and copula families, including Archimedean copulas, and establish stability under Markov product transformations.
As a consequence, we derive comparison results for measures of directed dependence arising from Markov structures, including Chatterjee’s rank correlation and a related sensitivity measure.
\end{abstract}

\noindent%

\section{Introduction}

Stochastic dependence among multiple random variables can be described in various ways. 
One common approach is to use measures of association, which are typically designed to capture specific forms of relationship, such as linear dependence (measured by Pearson correlation), monotone dependence (measured by Kendall’s tau and Spearman’s rho), or, more broadly, functional and directed dependence (measured by Chatterjee’s rank correlation \cite{chatterjee2020} and sensitivity measures such as the Sobol' index \cite{Sobol1993}).

A widely used alternative is to identify dependence patterns of varying strength, including positive lower orthant dependence, tail monotonicity, stochastic monotonicity, or multivariate total positivity of order two. These and related notions have been studied extensively \cite{karlin_1980, brindley_1972, fallat_total_2017, lehmann1966, sarkar_1969, nelsen_introduction_2007, joe_multivariate_1997, joe_dependence_2014, colangelo2005, fuchs2024JMVA_MT} because they can be highly informative in applications. 
For example, multivariate total positivity of order two for densities has important consequences for false discovery rate control in multiple testing procedures \cite{sarkar_2002} and for characterizing Markov structures \cite{fallat_total_2017}. Stochastic monotonicity, in turn, yields comparison results for Markov tree distributions \cite{ansari2025}. Moreover, tail monotonicity leads to comparison results for concordance measures such as Kendall's tau, Spearman's rho, and Spearman's footrule; see, for example, \cite{fredricks_relationship_2007} and Theorem \ref{Thm:FootruleRho} below, which establishes a new inequality between Spearman's rho and Spearman's footrule.

While (multivariate) total positivity of order two for distribution functions (\(\MTP\), for short) and densities (\(\dMTP\), for short) has been well studied \cite{karlin_1980, fallat_total_2017, muller2002}, the corresponding notion for conditional distributions (\(\cMTP\), for short), with the conditioning variables explicitly taken into account, has so far been examined only in the bivariate setting under the assumption of continuous marginal distribution functions \cite{fuchs_total_2023, caperaa1990,guillem2000}. 
For this particular setting, its relation to other notions of positive dependence is illustrated in \cite[Fig.~2]{fuchs_total_2023}.

In this paper, we extend the notion of total positivity of order two for conditional distributions to random vectors of dimension \(d > 2\): 
A function \(h: \bbR^d \to \bbR, d \geq 2,\) is said to be \emph{multivariate totally positive of order two} ($\mtp$, for short) if the inequality 
\begin{align} \label{MTP2}
  h(\bx) \, h(\by)
  & \leq h(\bx \land \by) \, h(\bx \lor \by)  
\end{align}
holds for all \(\bx,\by \in \bbR^d\), 
where $\bx \land \by$ and $\bx \lor \by$ denote the component-wise minimum and maximum, respectively.
If $d=2$, $h$ is usually referred to as \emph{totally positive of order two}.

The random vector \(\bX\) is said to be \cite{karlin_1980, nelsen_introduction_2007,  muller2002}
\begin{itemize}
\item 
\(\LOMTP(\bX)\) if the distribution function of \(\bX\), viewed as the mapping \(\bx \mapsto \bbP( \bX \leq \bx)\), is \(\mtp\);
\item 
\(\dMTP(\bX)\) if \(\bX\) has a density \(f_{\bX}\) with respect to the product probability measure \(\Motimes_{i=1}^d \bbP^{X_i}\) and \(f_{\bX}\) is \(\mtp\).
\end{itemize}

In Definition \ref{Def:cMTP2AB} and throughout the paper, we assume the existence of the regular conditional probabilities \(\bbP(\bZ_2 \leq \bz_2 \,|\, \bZ_1 = \bz_1)\) for all \(\bz_1 \in \bbR^{k}\) and all \(\bz_2 \in \bbR^{l}\), \(k,l \in \mathbb{N}\), and we refer to \cite{kallenberg_foundations_2021, klenke_probability_2007} for more background on regular conditional distributions.
Whenever a conditional distribution or density is used in an $\mtp$ statement, the assertion is understood to mean that there exist versions of the regular conditional distribution and of the density for which the displayed map satisfies the stated \(\mtp\) inequality.

\begin{definition}[An mtp$_2$ property for conditional distributions] \label{Def:cMTP2AB}\label{Def:cMTP2}~~
\begin{enumerate}[(1)]
\item Let \(A\) and \(B\) be two non-empty sets that form a partition of \(\{1,2,\dots,d\}\).
The random vector \(\bX = (\bX_A,\bX_B)\) is said to be \(\cLOMTP(\bX_B|\bX_A)\) if \(\bX_A\) has a density \(f_{\bX_A}\) with respect to the product measure 
\(\Motimes_{i \in A} \bbP^{X_i}\) and the mapping 
\begin{align}\label{Def:cMTP2AB:Eq}
   (\bx_A,\bx_B) \mapsto \bbP(\bX_B \leq \bx_B \,|\, \bX_A = \bx_A) \, f_{\bX_A} (\bx_A)
\end{align} 
is $\mtp$. 
\item The random vector \(\bX\) is said to be \(\cLOsMTP(\bX)\) 
if \(\bX\) is \(\cLOMTP(\bX_B|\bX_A)\) for all pairs of non-empty sets \(A, B\) that partition \(\{1,2,\dots,d\}\).
\end{enumerate}
\end{definition}
We emphasize that \(\cLOMTP(\bX_B|\bX_A)\) in \eqref{Def:cMTP2AB:Eq} does not require \(\bX\) to admit a density nor does it require the existence of a conditional density of \(\bX_B\) given \(\bX_A\).

The \(\mtp\) property for conditional distributions introduced in Definition \ref{Def:cMTP2AB} is weaker than total positivity of order two for densities (\(\dMTP\)), but stronger than the corresponding property for distribution functions (\(\MTP\)), and stronger than stochastic monotonicity (\(\LOSD\)); see Figures \ref{Fig:MTP2:Seq} and \ref{Fig:DepProp_1}. 
\SF{It therefore serves as a natural intermediate concept interpolating between \(\LOMTP(\bX)\) and \(\dMTP(\bX)\) and linking the different notions of positive dependence.}
\begin{figure}[h]
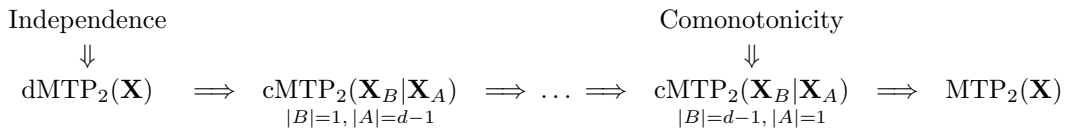
 
\begin{center}
$$
\begin{array}{cccccccccc}
\textrm{Independence} &&&&&& \textrm{Comonotonicity} &&&
\\
\Downarrow &&&&&& \Downarrow &&&
\\
\dMTP(\bX) & \Longrightarrow  
& \underset{|B|=1, \,|A|=d-1}{\cLOMTP(\bX_B|\bX_A)} & \Longrightarrow 
& \!\!\!\dots\!\!\! & \Longrightarrow & \underset{|B|=d-1, \,|A|=1}{\cLOMTP(\bX_B|\bX_A)} & \Longrightarrow & \LOMTP(\bX) 
\end{array}
$$
\end{center}
\caption{Interrelations among the different notions of multivariate total positivity of order two.
Arrows pointing from left to right indicate strict implications; that is, the property on the left is strictly stronger than, and therefore implies, the property on the right. 
The reverse directions in Figure \ref{Fig:MTP2:Seq} generally fail.}
\label{Fig:MTP2:Seq}
\end{figure}

We present conditions under which the \(\cMTP\) property holds for specific classes of distributions and copulas, focusing on binary distributions and Archimedean copulas.
We show that, for binary distributions, \(\cMTP\) is equivalent to \(\dMTP\), whereas for Archimedean copulas, \(\cMTP\) is equivalent to a suitable notion of stochastic monotonicity. 
Moreover, as in the bivariate case \cite{fuchs_total_2023}, these dependence properties for Archimedean copulas admit a complete characterization in terms of the Archimedean generator.\pagebreak

For practical applications, it is important to understand how positive dependence properties are transferred from random vectors to their associated Markov structures. We complement the well-known stability results for \(\dMTP\) \cite{fallat_total_2017} and stochastic monotonicity \cite{joe1994,Siburg-2021} by showing that \(\cMTP\) is likewise preserved under Markov product transformations. 
In addition, because measures of directed dependence such as Chatterjee's rank correlation \(\xi\) \cite{chatterjee2020} and the Sobol' index, also referred to as the fraction of explained variance \(R^2\), admit representations in terms of Markov structures \cite{fuchs2024JMVA}, these stability results directly imply comparison results for \(\xi\) and \(R^2\) based on the dependence properties of the underlying random vector.

\bigskip
The paper is organized as follows: 
Section \ref{Sec:MTP2} begins by studying positive dependence properties relative to a fixed partition of the random variables into conditioning and conditioned subsets, and examines the relationships among these properties.
It then considers positive dependence properties for random vectors without imposing a prescribed partition, providing a detailed analysis of the implications among the resulting concepts.
Section \ref{Sec:Ex} verifies the \(\cMTP\) property for several multivariate distributions and copula families.
Section \ref{Sec:markovproduct} establishes stability results for Markov structures and discusses their implications for comparison results involving measures of directed dependence.
Auxiliary results are collected in the Appendix (Section \ref{Sec:App}).

Throughout the paper, we focus on lower orthant dependence properties, that is, properties formulated in terms of (conditional) distribution functions. Analogous dependence properties can be defined in terms of survival functions, leading to corresponding notions of upper orthant dependence.

\section{Multivariate total positivity of order two for conditional distributions} 
\label{Sec:MTP2}

Let \(A\) and \(B\) be two non-empty sets that form a partition of \(\{1,2,\dots,d\}\), where \(d \geq 2\). For a \(d\)-dimensional random vector \(\bX = (X_1, \dots, X_d)\), we write \(\bX_A\) and \(\bX_B\) for the subvectors of \(\bX\) consisting of the coordinates indexed by $A$ and $B$, respectively.


\subsection{A new class of positive dependence properties}
\label{Sec:MTP2:A}

Multivariate total positivity of order two for densities and distribution functions are well established concepts in the literature \cite{karlin_1980, fallat_total_2017}. 
By contrast, to the best of our knowledge, multivariate total positivity of order two for conditional distributions that explicitly incorporate the conditioning variables has so far been investigated only in the bivariate case \cite{fuchs_total_2023, caperaa1990,guillem2000}. 
The definition of \(\cLOMTP(\bX_B|\bX_A)\) in \eqref{Def:cMTP2AB:Eq} is in line with the existing literature as it is a genuine generalization of the bivariate concept studied in \cite{fuchs_total_2023}.

\begin{remark} \label{Rem.cMTP2:Ex}
\begin{enumerate}[{\rm (1)}]
\item \label{Rem.cMTP2:Ex.1}
\(\cLOMTP(\bX_B|\bX_A)\) implies that the conditional distribution function, viewed as a mapping 
\begin{align}\label{Def:cMTP2:Eq2}
    \bx_B \mapsto \bbP(\bX_B \leq \bx_B \,|\, \bX_A = \bx_A) 
\end{align} 
is \(\mtp\) for almost all \(\bx_A\).
However, because \(\cLOMTP(\bX_B|\bX_A)\) in \eqref{Def:cMTP2AB:Eq} also incorporates the conditioning variables indexed by \(A\), \(\cLOMTP(\bX_B|\bX_A)\) is strictly stronger than \eqref{Def:cMTP2:Eq2}. 

\item 
Recall that the standard definition of $\mtp$ for densities requires the existence of a density with respect to the product measure; see, for example, \cite{karlin_1980, muller2002}.
When translated to the setting of conditional distributions in \eqref{Def:cMTP2AB:Eq}, this corresponds to assuming that \(\bX_A\) admits a density \(f_{\bX_A}\) with respect to \(\Motimes_{i \in A} \bbP^{X_i}\).
\SF{We refer to Section \ref{Sec:Dis} for a discussion of density-free analogues of \(\dMTP\) and \(\cLOMTP\).}

\item \label{Rem.cMTP2:Ex:2} 
If \(|A|=1\), then 
\(\cLOMTP(\bX_B|\bX_A)\) if and only if 
the mapping
\begin{align}\label{Def:cMTP2:Eq}
   (x,\bx_B) \mapsto \bbP(\bX_B \leq \bx_B \,|\, \bX_A = x) 
\end{align} 
is $\mtp$. 
\SF{Although every random vector \(\bX\) with independent components, as well as every random vector \(\bX\) with comonotone components, satisfies a multivariate analogue of \eqref{Def:cMTP2:Eq}, i.e.}
\begin{align}\label{Def:cMTP2:Eq3}
   (\bx_A,\bx_B) \mapsto \bbP(\bX_B \leq \bx_B \,|\, \bX_A = \bx_A) 
\end{align} 
is $\mtp$ for any choice of \(A\) and \(B\) forming a partition of \(\{1,\dots,d\}\), this concept is, in general, too restrictive. 
In particular, one would like an \(\mtp\) property of the conditional distribution function of \(\bX_B\) given \(\bX_A\) to follow from \(\dMTP(\bX)\).
This implication, however, fails for \eqref{Def:cMTP2:Eq3} and arbitrary choices of \(A\), as demonstrated in Example \ref{Ex:GaussianMulti}.
In contrast, we show in Theorem \ref{Thm.MTP2} below that \(\dMTP(\bX)\) implies \(\cLOMTP(\bX_B|\bX_A)\) for any choice of \(A\) and \(B\) forming a partition of \(\{1,\dots,d\}\).
\end{enumerate}
\end{remark}

\SF{The following example highlights the importance of including the density \(f_{\bX_A}\) as a required factor in the definition of \(\cMTP\) in \eqref{Def:cMTP2AB:Eq}.}

\begin{example} \label{Ex:GaussianMulti}
Let \(\bX = (X_1,X_2, X_3) \sim \cN(\bZERO,\Sigma)\) be a \(3\)--dimensional normally distributed random vector with covariance matrix
\begin{align*}
\Sigma=\begin{bmatrix}
1 & 0.5 & 0.2\\
0.5 & 1 & 0.3\\
0.2 & 0.3 & 1
\end{bmatrix}.
\end{align*}
Then, for \(A=\{1,2\}\) and \(B=\{3\}\), the map in \eqref{Def:cMTP2:Eq3} is not \(\mtp\), but \(\dMTP(\bX)\) holds.
\begin{enumerate}[{\rm (1)}]
\item 
Consider the partition \(A = \{1,2\}\) and \(B = \{3\}\). 
According to \cite[Corollary 5]{Cambanis1981}, the conditional distribution of \(X_B\) given \(\bX_A = \bx_A\) is again normally distributed with mean vector $\Sigma_{AB}^T \Sigma_{AA}^{-1} \, \bx_A $ and covariance matrix $\Sigma_{BB} -\Sigma_{AB}^{T}\Sigma_{AA}^{-1}\Sigma_{AB}$, where \(\Sigma _{BB} = 1\), \(\Sigma _{BA} = (0.2,0.3)\), and \(\Sigma _{AA} = \begin{bmatrix}
1 & 0.5\\
0.5 & 1
\end{bmatrix}\).
Choosing \((\bx_A,x_B) = (-1,1,1)\) and \((\by_A,y_B) = (1,-1,1)\), straightforward calculation  yields
\begin{align*}
  &\bbP(X_B \leq x_B \,|\, \bX_A = \bx_A)
   = \bbP(X_3 \leq 1 \,|\, X_1 = -1, X_2 = 1) 
  \approx  0.7996,
  \\
 & \bbP(X_B \leq y_B \,|\, \bX_A = \by_A)
   = \bbP(X_3 \leq 1 \,|\, X_1 = 1, X_2 = -1) 
  \approx  0.8962,
  \\ 
  &\bbP(X_B \leq x_B \wedge y_B \,|\, \bX_A = \bx_A \wedge \by_A)
   = \bbP(X_3 \leq 1 \,|\, X_1 =-1, X_2 = -1) 
  \approx  0.9193, 
  \\
  &\bbP(X_B \leq x_B \lor y_B \,|\, \bX_A = \bx_A \lor \by_A)
  = \bbP(X_3 \leq 1 \,|\, X_1 = 1, X_2 = 1)  
  \approx  0.7581,
\end{align*}
and hence 
\begin{align*}
& \bbP(X_B \leq x_B \,|\, \bX_A = \bx_A) \, 
\bbP(X_B \leq y_B \,|\, \bX_A = \by_A)
   = 0.7166 
\\
&  > 0.6969 
   = \bbP(X_B \leq x_B \land y_B \,|\, \bX_A = \bx_A \land \by_A) \,
\bbP(X_B \leq x_B \lor y_B \,|\, \bX_A = \bx_A \lor \by_A) \,.
\end{align*}
In other words, \eqref{Def:cMTP2:Eq3} does not hold.

\item
Moreover, \(\Sigma^{-1}\) is an \(M\)-matrix \cite{fallat_total_2017, karlin_m-matrices_1983}, that is, \(\Sigma^{-1}_{ij} \leq 0\) for all \(i \neq j\), \emph{and} \(\Sigma^{-1}_{ii} > 0\) for all \(i \in \{1,2,3\}\).
It is well-known that \(\dMTP(\bX)\) holds if and only if \(\Sigma^{-1}\) is an \(M\)-matrix; see \cite[Example 3.1]{karlin_1980} or \cite[Theorem 3]{Ru1981}.
\end{enumerate}
\end{example}

We now establish a replacement property, which will be useful throughout and allows the positive dependence properties defined in \eqref{Def:cMTP2AB:Eq} to be ordered by strength.

\begin{theorem} \label{cMTP2:Seq}
Let \(A\) and \(B\) be two non-empty sets that form a partition of \(\{1,2,\dots,d\}\) with \(|A| \geq 2\).
Then, for every \(k \in A\),
\begin{align*}
  \cLOMTP(\bX_B|\bX_A) \quad \textrm{ implies } \quad \cLOMTP(\bX_{B \cup \{k\}}|\bX_{A \backslash \{k\}}).
\end{align*}
\end{theorem}
\begin{proof}
Set \(A' := A \backslash \{k\}\) and \(B' := B \cup \{k\}\).
\(\cLOMTP(\bX_B|\bX_A)\) implies that \(\bX_{A'}\) has \(\Motimes_{i \in A'} \bbP^{X_i}\)--density 
\begin{align*}
    f_{\bX_{A'}} (\bx_{A'})
    & = \int_{\bbR} f_{\bX_{A}} (\bx_{A'},z_{k}) \de \bbP^{X_{k}}(z_{k})\,,
\end{align*}
and disintegration yields
\begin{align*}
  & \bbP(\bX_{B'} \leq \bx_{B'} \,|\, \bX_{A'} = \bx_{A'}) \, f_{\bX_{A'}} (\bx_{A'})
  \\
  & = \bbP(X_k \leq x_k, \bX_{B} \leq \bx_{B} \,|\, \bX_{A'} = \bx_{A'}) \, f_{\bX_{A'}} (\bx_{A'})
  \\
  & = \int_{(-\infty,x_k]} \bbP(\bX_{B} \leq \bx_{B} \,|\, \bX_{A'} = \bx_{A'}, X_{k} = z_{k}) \, f_{\bX_{A}} (\bx_{A'},z_{k}) \de \bbP^{X_{k}} (z_{k})
  \\*
  & = \int_{(-\infty,x_k]} \bbP(\bX_{B} \leq \bx_{B} \,|\, \bX_{A} = (\bx_{A'},z_k)) \, f_{\bX_{A}} (\bx_{A'},z_k) \de \bbP^{X_{k}} (z_{k})\,.
\end{align*}
Since the mapping \((\bx_{A},\bx_B) \mapsto \bbP(\bX_{B} \leq \bx_{B} \,|\, \bX_{A} = \bx_{A}) \, f_{\bX_{A}} (\bx_{A})\) is $\mtp$ by assumption, \(\cLOMTP(\bX_{B'}|\bX_{A'})\) now follows from Lemma \ref{MTP:Help1}. 
\end{proof}

As a consequence of Theorem \ref{cMTP2:Seq}, the \(\cLOMTP(\bX_B|\bX_A)\) property is closed under marginalization in both \(A\) and \(B\).

\begin{corollary}(Closure under marginalization)~~ \label{Cor.cMTP2:Margins}
\(\cLOMTP(\bX_B|\bX_A)\) implies \(\cLOMTP(\bX_{B'}|\bX_{A'})\) for any choice of \(A' \subseteq A\) and \(B' \subseteq B\) with \(A', B'\) non-empty.
\end{corollary}
\begin{proof}
It is clear from the definition that \(\cLOMTP(\bX_B|\bX_A)\) is closed under marginalization
in \(B\).
Closure under marginalization in \(A\) directly follows from Theorem \ref{cMTP2:Seq} and the corresponding property in \(B\).
\end{proof}

We now establish the relationship between \(\cLOMTP(\bX_B|\bX_A)\) and the other notions of multivariate total positivity of order two.

\begin{theorem} \label{Thm.MTP2}
Let \(A\) and \(B\) be two non-empty sets that form a partition of \(\{1,2,\dots,d\}\).
Then
\begin{enumerate}[{\rm (1)}]
\item \label{Thm.MTP2:3} \(\dMTP(\bX)\) implies \(\cLOMTP(\bX_B|\bX_A)\).
\item \label{Thm.MTP2:1} \(\cLOMTP(\bX_B|\bX_A)\) implies \(\LOMTP(\bX)\).
\end{enumerate}
\end{theorem}
\begin{proof}
We first prove \eqref{Thm.MTP2:3}.
\(\dMTP(\bX)\) yields the existence of a \(\Motimes_{i=1}^d \bbP^{X_i}\)--density \(f_{\bX}\) such that \(\bX_{A}\) has \(\Motimes_{i \in A} \bbP^{X_i}\)--density 
\begin{align*}
    f_{\bX_{A}} (\bx_{A})
    & = \int_{\bbR^{|B|}} f_{\bX} (\bx_{A},\bx_{B}) \de \Motimes_{i \in B} \bbP^{X_i}(x_{i})\,.
\end{align*}
Disintegration further gives
\begin{align*}
  \bbP(\bX_B \leq \bx_B \, | \, \bX_A = \bx_A) \, f_{\bX_A} (\bx_A)
  &   =  \int_{(-\boldsymbol{\infty},\bx_B]} f_{\bX} (\bx_{A},\bz_{B}) \de \Motimes_{i \in B} \bbP^{X_i}(z_{i})\,.
\end{align*}
Since \(f_{\bX}\) is \(\mtp\) by assumption, the assertion follows from Lemma \ref{MTP:Help1}.
\\
We now prove \eqref{Thm.MTP2:1}.
\(\cLOMTP(\bX_B|\bX_A)\) yields the existence of a \(\Motimes_{i \in A} \bbP^{X_i}\)--density \(f_{\bX_{A}}\) and disintegration gives
\begin{align*}
  \bbP(\bX_A \leq \bx_A, \bX_B \leq \bx_B) 
  &   =  \int_{(-\boldsymbol{\infty},\bx_A]} \bbP(\bX_B \leq \bx_B \, | \, \bX_A = \bz_A) \de \bbP^{\bX_A}(\bz_{A})
  \\
  &   =  \int_{(-\boldsymbol{\infty},\bx_A]} \bbP(\bX_B \leq \bx_B \, | \, \bX_A = \bz_A) f_{\bX_{A}}(\bz_{A}) \de \Motimes_{i \in A} \bbP^{X_i}(z_{i})\,.
\end{align*}
Again, since \((\bx_A,\bx_B) \mapsto  \bbP(\bX_B \leq \bx_B \, | \, \bX_A = \bx_A) f_{\bX_{A}}(\bx_{A})\) is \(\mtp\) by assumption, the assertion now follows from Lemma \ref{MTP:Help1}.
This proves \eqref{Thm.MTP2:1}. 
\end{proof}

Figure \ref{Fig:MTP2:Seq} illustrates the findings in Theorem \ref{cMTP2:Seq} and Theorem  \ref{Thm.MTP2}. 
We show in Section \ref{Subsect:Archimedean} that the reverse directions in Figure \ref{Fig:MTP2:Seq} generally fail.
By Remark \ref{Rem.cMTP2:Ex}, \(\cLOMTP(\bX_{B}|\bX_{A})\) with \(|A| = 1\) is the strongest among the positive dependence properties in Figure \ref{Fig:MTP2:Seq} that do not require the existence of a density and are satisfied by \SF{random vectors \(\bX\) with comonotone components}.

\subsection{\texorpdfstring{\(\cMTP\) and its relation to classical positive dependence properties referring to a partition}{cLOMTP and its relation to classical positive dependence properties referring to a partition}} \label{Sec:MTP2.B}

We now relate \(\cLOMTP(\bX_B|\bX_A)\) to two well-established concepts from the literature: tail monotonicity and stochastic monotonicity. 
The random vector \(\bX = (\bX_A,\bX_B)\) is said to be \cite{brindley_1972, nelsen_introduction_2007, Ru1981, harris_1970}
\begin{itemize}
\item 
\(\LTD(\bX_B|\bX_A)\), \emph{left tail decreasing in \(A\)}, if the mapping (whenever it exists)
\begin{equation} \label{LTD}
\bx_A \mapsto \bbP(\bX_B \leq \bx_B \,|\, \bX_A \leq \bx_A) \text{ is non-increasing for all } \bx_B \in \bbR^{|B|}\,.
\end{equation}

\item 
\(\LOSD (\bX_B|\bX_A)\), 
if for (a version of) the regular conditional distribution of \(\bX_B\) given \(\bX_A\) the mapping
\begin{equation} \label{LOSI}
\bx_A \mapsto \bbP(\bX_B \leq \bx_B \,|\, \bX_A = \bx_A) \text{ is non-increasing for all } \bx_B \in \bbR^{|B|} \text{ and } \bbP^{\bX_A}\textrm{-almost all } \bx_A \in \bbR^{|A|}.
\end{equation}
For singletons $B = \{i\}$, \(\bx_A \mapsto \bbP(X_B > x_B\,|\, \bX_A = \bx_A) = 1- \bbP(X_B \leq x_B\,|\, \bX_A = \bx_A) \) is non-decreasing and the property is usually referred to as \emph{stochastic increasingness} \cite{muller2006, Ru2013}. 
\end{itemize}



\begin{remark}\label{Cor.SI:Margins}
\begin{enumerate}[{\rm (1)}]
\item 
It is immediate from the definition that the $\LTD(\bX_B|\bX_A)$ property is closed under marginalization in both \(A\) and \(B\).
\item \label{Cor.SI:Margins:2}
Instead, while \(\LOSD (\bX_B|\bX_A)\) is closed under marginalization in \(B\), it is not closed under marginalization in \(A\); see Example \ref{Ex:SI:Margins}.
\end{enumerate}
\end{remark}

\begin{example} \label{Ex:SI:Margins}
Suppose \(\bX = (X_1,X_2,X_3) \in \{0,1\}^3\) with probabilities of occurrence
\begin{center}
$\displaystyle \begin{array}{ c|c c c c c c c c }
\bx = (x_{1},x_{2},x_{3}) & ( 0,0,0) & ( 1,0,0) & ( 0,1,0) & ( 0,0,1) & ( 1,1,0) & ( 1,0,1) & ( 0,1,1) & ( 1,1,1)\\
\hline
\bbP(\bX = \bx) & 0.09 & 0.28 & 0.12 & 0.01 & 0.02 & 0.12 & 0.28 & 0.08
\end{array}$
\end{center}
Consider \(A = \{1,2\}\), \(A' = \{1\}\), and \(B=\{3\}\). 
Then,  
\begin{center}
$\displaystyle \begin{array}{ c| c c c c }
(x_{1},x_{2}) & ( 0,0)  & ( 1,0) & ( 0,1) & ( 1,1) 
\\
\hline
\bbP(X_1 = x_{1},X_2=x_{2}) & 0.1  & 0.4 & 0.4 & 0.1 \\
\end{array}$
\end{center}
and 
\begin{center}
$\displaystyle \begin{array}{ c| c c c c }
(x_{1},x_{2}) & ( 0,0)  & ( 1,0) & ( 0,1) & ( 1,1) 
\\
\hline
\bbP(X_3 \leq 0 \,|\, (X_1,X_2)=(x_{1},x_{2})) & 0.9  & 0.7 & 0.3 & 0.2 \\
\bbP(X_3 \leq 1 \,|\, (X_1,X_2)=(x_{1},x_{2})) & 1  & 1 & 1 & 1 \\
\end{array}$
\end{center}
which proves $\LOSD(\bX_B|\bX_A)$.
In contrast, 
\(\bbP(X_3 \leq 0 \,|\, X_1=0)  = 0.42 < 0.6 = \bbP(X_3 \leq 0 \,|\, X_1=1)\),
which contradicts $\LOSD(\bX_B|\bX_{A'})$.
\end{example}


The next result relates \(\cLOMTP(\bX_B|\bX_A)\) to stochastic monotonicity \(\LOSD(\bX_B|\bX_A)\) and tail monotonicity \(\LTD(\bX_B|\bX_A)\).
Assertion~\eqref{Prop.SI.LTD.RTI:1} of Theorem~\ref{Thm.MTP2.SI} is well established in the literature \cite{lehmann1966}.

\begin{theorem}\label{Prop.SI.LTD.RTI}\label{Thm.MTP2.SI}
We have
\begin{enumerate}[{\rm (1)}]
\item \label{Thm.MTP2.SI:1} \(\cLOMTP(\bX_B|\bX_A)\) implies \(\LOSD(\bX_B|\bX_A)\).
\item \label{Thm.MTP2.SI:2} \SF{\(\cLOMTP(\bX_B|\bX_A)\) implies \(\LTD(\bX_B|\bX_A)\).}
\item \label{Prop.SI.LTD.RTI:1} \SF{If \(|A|=1\), then \(\LOSD(\bX_B|\bX_A)\) implies \(\LTD(\bX_B|\bX_A)\).}
\end{enumerate}
\end{theorem}
\begin{proof}
We first prove \eqref{Thm.MTP2.SI:1}.
Consider \(\bx_B \in \bbR^{|B|}\).
Then \(\cLOMTP(\bX_B|\bX_A)\) yields the existence of a \(\Motimes_{i \in A} \bbP^{X_i}\)--density \(f_{\bX_A}\) such that, for almost all \(\bx_A,\by_A \in \bbR^{|A|}\) with $\bx_A \leq \by_A$,
\begin{align*}
  & \bbP(\bX_B \leq \bx_B \, | \, \bX_A = \by_A) \, f_{\bX_A} (\by_A) \, f_{\bX_A} (\bx_A)
  \\
  &   =  \bbP(\bX_B \leq \bx_B \, | \, \bX_A = \by_A) \, f_{\bX_A} (\by_A) \, \bbP(\bX_B \in \bbR^{|B|} \, | \, \bX_A = \bx_A) \, f_{\bX_A} (\bx_A)
  \\
  & \leq \bbP(\bX_B \leq \bx_B \, | \, \bX_A = \bx_A) \, f_{\bX_A} (\bx_A) \, \bbP(\bX_B \in \bbR^{|B|} \, | \, \bX_A = \by_A) \, f_{\bX_A} (\by_A)
  \\
  &   =  \bbP(\bX_B \leq \bx_B \, | \, \bX_A = \bx_A) \, f_{\bX_A} (\bx_A) \, f_{\bX_A} (\by_A)\,.
\end{align*}
The \(\mtp\) inequality is achieved by taking a sequence \(\by_B^{(n)} \uparrow (\infty, \dots,\infty)\) and passing to the limit.
For all $\bx_A \leq \by_A$ with $f_{\bX_A}(\bx_A) f_{\bX_A}(\by_A)>0$, division yields the desired inequality.
This proves \(\LOSD(\bX_B|\bX_A)\). 
\\
\SF{Again, assume that \(\cLOMTP(\bX_B|\bX_A)\) holds. 
Then Theorem \ref{Thm.MTP2}\eqref{Thm.MTP2:1} implies \(\LOMTP(\bX)\).
Now, choose two non-empty sets \(A, B\) partitioning \(\{1,2,\dots,d\}\), and note that \(\LTD(\bX_B | \bX_A)\) is equivalent to 
\begin{align*}
  \bbP(\bX_A \leq \by_{A}, \bX_B \leq \bx_B) \, \bbP (\bX_A \leq \bx_A, \bX_B \in \bbR^{|B|})
  & \leq \bbP(\bX_A \leq \bx_{A}, \bX_B \leq \bx_B) \, \bbP (\bX_A \leq \by_{A}, \bX_B \in \bbR^{|B|})
\end{align*}
for all \(\bx_B \in \bbR^{|B|}\) and all \(\bx_A, \by_{A} \in \bbR^{|A|}\) with \(\bx_A \leq \by_{A}\) such that $\bbP( \bX_A \leq \bx_A) > 0$.
But this inequality is indeed fulfilled since \(\LOMTP(\bX)\) holds; the \(\mtp\) inequality is achieved by taking a sequence \(\by_B^{(n)} \uparrow (\infty, \dots,\infty)\) and passing to the limit.
\\
Assertion \eqref{Prop.SI.LTD.RTI:1} is due to \cite{lehmann1966}.}
\end{proof}

Figure \ref{Fig:ConditionalDepProp} illustrates the findings in Theorem~\ref{Thm.MTP2.SI}.
It follows from \cite[Example 3.5]{fuchs_total_2023}
that the reverse direction in Theorem~\ref{Thm.MTP2.SI}\eqref{Thm.MTP2.SI:1} does not hold, in general. A counterexample for the reverse direction in Theorem~\ref{Thm.MTP2.SI}\eqref{Prop.SI.LTD.RTI:1} can be found in \cite[Remark 2.7]{fuchs_total_2023}.
\begin{figure}[h]
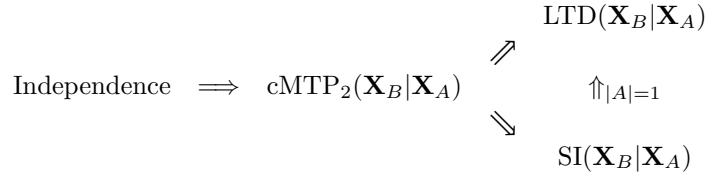
 
\begin{center}
$$
\begin{array}{cccccc}
&&&& \LTD(\bX_B|\bX_A) 
\\
&&& \Nearrow &
\\
\textrm{Independence} & \Longrightarrow  & \cLOMTP(\bX_B|\bX_A) && \Uparrow_{|A|=1}
\\
&&& \Searrow &
\\
&&&& \LOSD(\bX_B|\bX_A)
\end{array}
$$
\end{center}
\caption{Interrelations among the different positive dependence properties discussed in Subsection \ref{Sec:MTP2.B} that refer to a given partition of \(\{1,2,\dots,d\}\), \(d \geq 2\), into non-empty sets \(A\) and \(B\).
Arrows pointing from left to right indicate strict implications; that is, the property on the left is strictly stronger than, and therefore implies, the property on the right.}
\label{Fig:ConditionalDepProp}
\end{figure}

\subsection{\texorpdfstring{\(\cMTP\) and its relation to classical positive dependence properties of random vectors}{cLOMTP and its relation to classical positive dependence properties of random vectors}} \label{Sec:PDP.RV}

The dependence properties introduced in Subsection \ref{Sec:MTP2.B} for a fixed partition now serve as a basis for formulating and analyzing the corresponding properties for random vectors \(\bX = (X_1,X_2, \dots, X_d)\), where \(d \geq 2\), without imposing any particular partition on its components.

The random vector \(\bX\) is said to be \cite{karlin_1980, brindley_1972, nelsen_introduction_2007,  muller2002}
\begin{itemize}
\item 
$\PLOD(\bX)$, \emph{positive lower orthant dependent}, if the inequality 
\(\bbP( \bX \leq \bx) \geq \prod^d_{i = 1} \bbP(X_i \leq x_i)\) holds for all \(\bx \in \bbR^d\); 
\item 
$\LTD(\bX)$, \emph{left tail decreasing}, if \(\bX\) is \(\LTD(\bX_B|\bX_A)\) for all pairs of non-empty sets \(A, B\) partitioning \(\{1,2,\dots,d\}\);
\item 
\(\LOSD(\bX)\), \emph{(lower orthant) stochastically increasing}, 
if \(\bX\) is \(\LOSD (\bX_B|\bX_A)\) for all pairs of non-empty sets \(A, B\) partitioning \(\{1,2,\dots,d\}\).
\end{itemize}
Recall the definitions of \(\MTP(\bX)\) and \(\cMTP(\bX)\) from the introduction.

\begin{remark}
Related notions of stochastic monotonicity that refer to specific subsets \(A\) and \(B\), not necessarily a partition of \(\{1,\dots,d\}\), are listed below: 
\(\bX\) is said to be \cite{nelsen_introduction_2007, joe_multivariate_1997, muller2006, Ru2013}
\begin{itemize}
\item 
\(\LOPDS(\bX)\), \emph{lower orthant positive dependent through the stochastic ordering},  
if the vector \(\bX\) is \(\LOSD (\bX_B|\bX_A)\) for all singletons \(A=\{i\}\) and \(B = \{1,2,\dots,d\}\backslash\{i\}\), \(i \in \{1,2,\dots,d\}\); 
\item 
\(\CIS(\bX)\), \emph{conditionally increasing in sequence}, 
if the vector \(\bX\) is \(\LOSD (\bX_B|\bX_A)\) for all singletons $B=\{i\}$ and $A = \{1,2,\dots,i-1\}$, \(i \in \{2,\dots,d\}\); 
\item 
\(\CI(\bX)\), \emph{conditionally increasing}, 
if the vector \(\bX\) is \(\LOSD (\bX_B|\bX_A)\) for all singletons $B=\{i\}$ and \(\emptyset \neq A \subseteq \{1,2,\dots, d\}\backslash\{i\}\), \(i \in \{1,2,\dots,d\}\).
\end{itemize}
We note in passing that in the bivariate case ($d=2$) \(\LOSD(\bX)\), \(\LOPDS(\bX)\), and \(\CI(\bX)\) are equivalent. 
If, in addition, the bivariate random vector \(\bX\) is exchangeable, then these properties are also equivalent to \(\CIS(\bX)\). 
\end{remark}


\begin{remark} \label{ClosedMargins.2}
It is immediate from the definitions that $\PLOD(\bX)$, $\LTD(\bX)$, $\LOSD(\bX)$, and \(\LOMTP(\bX)\) are closed under marginalization and, 
due to \cite[Lemma 2.12]{sarkar_1969} (see also \cite[Proposition 3.2]{karlin_1980}), the same applies to \(\dMTP(\bX)\). Closure under marginalization for \(\cLOsMTP(\bX)\) follows from Corollary \ref{Cor.cMTP2:Margins}. 
\end{remark}

The next result is due to Theorem~\ref{Thm.MTP2.SI}\eqref{Thm.MTP2.SI:1} and Theorem \ref{Thm.MTP2}.

\begin{corollary} \label{Cor.MTP2}~~
We have
\begin{enumerate}[{\rm (1)}]
\item \label{Cor.MTP2:1} \(\dMTP(\bX)\) implies \(\cLOMTP(\bX)\).
\item \label{Cor.MTP2:2} \(\cLOsMTP(\bX)\) implies \(\LOSD(\bX)\).
\item \label{Cor.MTP2:3} \(\cLOMTP(\bX)\) implies \(\LOMTP(\bX)\).
\end{enumerate}
\end{corollary}


We further briefly summarize the connections between the classical positive dependence properties.

\begin{corollary}\label{Cor.LTD.RTI.POD}\label{Prop.SI.LTD.RTI:B}\label{Prop.PDS.PLOD:B}
We have
\begin{enumerate}[(1)]
\item \label{Prop.MTP2.LTD.RTI:1} \(\LOMTP(\bX)\) implies \(\LTD(\bX)\). 
\item \label{Cor.LTD.RTI.POD1}    $\LTD(\bX)$ implies $\PLOD(\bX)$. 
\end{enumerate}
\end{corollary}
\begin{proof}
A proof for assertion \eqref{Prop.MTP2.LTD.RTI:1} has been presented in the proof of Theorem \ref{Thm.MTP2.SI}\eqref{Thm.MTP2.SI:2}. 
The analogous assertion for survival functions instead of distribution functions is given in \cite[Theorem 3.1]{brindley_1972}.
Assertion \eqref{Cor.LTD.RTI.POD1} is due to \cite[Section 4]{sarkar_1969}.
\end{proof}

\begin{figure}[h!]
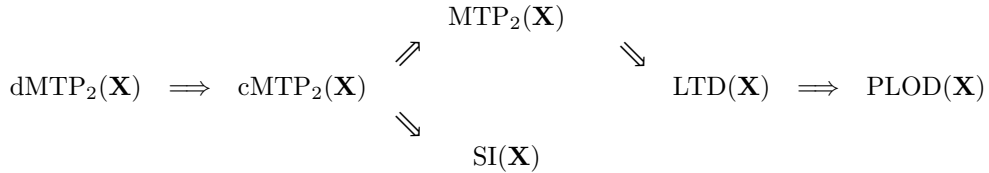
 
\begin{center}
$$
\begin{array}{cccccccccc}
&&&& \LOMTP(\bX)  &&&&&
\\
&&& \Nearrow &&& \Searrow &&&
\\
\dMTP(\bX) & \Longrightarrow  & \cLOsMTP(\bX) &&&&& \LTD(\bX) & \Longrightarrow & \PLOD(\bX)
\\
&&& \Searrow &&&  &&&
\\
&&&& \SI(\bX) &&&&&
\end{array}
$$
\end{center}
\caption{Interrelations among the different positive dependence properties discussed in Subsection \ref{Sec:PDP.RV}, with particular emphasis on the placement of \(\cLOsMTP\) within this hierarchy. 
Arrows pointing from left to right indicate strict implications; that is, the property on the left is strictly stronger than, and therefore implies, the property on the right.}
\label{Fig:DepProp_1}
\end{figure}

Figure \ref{Fig:DepProp_1} illustrates the findings of Section \ref{Sec:PDP.RV}.
It is well known that, in general, the reverse implications in Corollary \ref{Cor.LTD.RTI.POD} do not hold. Likewise, the failure of the reverse implications in Corollary \ref{Cor.MTP2} follows from the corresponding results in Subsections  \ref{Sec:MTP2:A} and \ref{Sec:MTP2.B}.
For $d=2$, Figure \ref{Fig:DepProp_1} reduces to \cite[Fig.~2]{fuchs_total_2023}.

\begin{remark} (Scale invariance) \label{Rem:Copula}~~
For a random vector \(\bX\) with continuous joint distribution function \(F_\bX\), Sklar's theorem \cite{nelsen_introduction_2007,durante_principles_2015} guarantees the existence of a \emph{unique} copula \(C_\bX\) that fully encodes the dependence structure among the components of \(\bX\), that is,
\begin{align*}
    F_\bX(\bx) = C_\bX (F_{X_1}(x_1), \dots, F_{X_d}(x_d))
\end{align*}
for all \(\bx \in \bbR^d\).
Consequently, provided the required densities exist where needed, any positive dependence property considered in this section can be equivalently verified at the copula level. 
This follows by applying the probability integral transform \cite{Ru2013}
\begin{align*}
  \bU := (F_{X_1}(X_1), \dots, F_{X_d}(X_d))\,,
\end{align*}
whose joint distribution function is the copula \(C_\bX\) and whose marginals are uniform on \([0,1]\).
\SF{This holds, in particular, for \(\cLOMTP(\bX_B|\bX_A)\), as in the case of continuous marginal distribution functions \(F_{X_i}\), \(i \in \{1,\dots,d\}\), \eqref{Def:cMTP2AB:Eq} simplifies to
\begin{align*}
    \bbP(\bX_B \leq \bx_B \,|\, \bX_A = \bx_A) \, f_{\bX_A} (\bx_A)
    & = \bbP \left( \bigcap_{k \in B} \{U_k \leq F_{X_k}(x_k)\} \,\Big|\, \bigcap_{i \in A} \{U_i = F_{X_i}(x_i)\} \right) \, c_{A} ({\bf F}_{\bX_A}(\bx_A))\,,
\end{align*}
where \({\bf F}_{\bX_A}(\bx_A) = (F_{X_i}(x_i))_{i \in A}\) and it is required that the copula of \(\bX_A\) has density \(c_A\).
Therefore, \(\bX\) is \(\cLOMTP(\bX_B|\bX_A)\) if the mapping 
\begin{align*}
    (\bu_A,\bu_B) 
    & \mapsto \bbP \left(\bU_B \leq \bu_B \,|\, \bU_A = \bu_A \right) \, c_{A} (\bu_A)
\end{align*}
is \(\mtp\).}
Hence, for continuous random vectors \(\bX\), the positive dependence properties under consideration are margin-free and may equivalently be formulated in terms of random vectors supported on \([0,1]^d\) with uniform marginals.
\end{remark}

\section{\texorpdfstring{Verifying $\cLOMTP(\bX)$ for families of multivariate distributions}{Verifying cLOMTP(X) in families of multivariate distributions}}\label{Sec:Ex}

We now investigate the dependence property \(\cLOMTP(\bX)\), introduced in Definition \ref{Def:cMTP2}, for selected families of multivariate distributions. We also clarify its relationship with the positive dependence properties displayed in Figure \ref{Fig:DepProp_1}.

\SF{Every random vector \(\bX\) with independent components} satisfies \(\dMTP(\bX)\) and \SF{every random vector \(\bX\) with comonotone components} satisfies \(\cLOsMTP(\bX_B|\bX_A)\) with \(|A|=1\) (Remark \ref{Rem.cMTP2:Ex}).
We now turn to mixtures of these two extremal cases; we refer to \cite[Example 3.5]{fuchs_total_2023} for the case \(d=2\).

\begin{example}[Mixtures of independence and comonotonicity]~~ \label{Ex:Mixture}
For \(\alpha \in [0,1]\) and univariate continuous distribution functions \(F_i\), \(i \in \{1,\dots,d\}\), \(d \geq 3\), consider the distribution function \(F_{\alpha}\) given by
\begin{align*}
  F_{\alpha} 
  & = \alpha F^{\textrm{cm}} + (1-\alpha) F^{\perp}\,,
\end{align*}
with \(F^{\textrm{cm}}(\bx) := \min\{F_1(x_1), \dots,F_d(x_d)\}\) and \(F^{\perp}(\bx) := \prod_{i=1}^d F_i(x_i)\).
Denote by \(\bX_\alpha\) a random vector with distribution function \(F_\alpha\). 
Then 
\begin{align*}
\begin{array}{ccccccccc}
  \dMTP(\bX_\alpha) 
  & \Longleftrightarrow & \cLOsMTP(\bX_\alpha) 
  & \Longrightarrow & \underset{|A|=1}{\cLOMTP((\bX_\alpha)_B|(\bX_\alpha)_A)}
  & \Longrightarrow & \underset{|A|=1}{\LOSD(\bX_\alpha)} 
  & \Longleftrightarrow & \LOMTP(\bX_\alpha) 
\end{array}  
\end{align*}
with the reverse implications being invalid.
More precisely, 
\begin{enumerate}[{\rm (1)}]
\item
\(\dMTP(\bX_\alpha)\), if and only if \(\alpha=0\);
\item 
\(\cLOMTP(\bX_\alpha)\), if and only if \(\alpha=0\);
\item 
\(\cLOMTP((\bX_\alpha)_B|(\bX_\alpha)_A)\) with \(|A|=1\), if and only if \(\alpha \in \{0,1\}\);
\item 
\(\LOSD(\bX_\alpha)\) with \(|A|=1\) for all \(\alpha \in [0,1]\).
\item 
\(\MTP(\bX_\alpha)\) for all \(\alpha \in [0,1]\).
\end{enumerate}
\end{example}

In the following example, we verify the \(\cLOMTP(\bX)\) property for binary distributions with strictly positive probabilities by applying Corollary \ref{Cor.MTP2}.
Here, the situation differs from the bivariate case, where \(\dMTP(\bX)\) and \(\PLOD(\bX)\) are equivalent, and changes considerably in the more general setting.

\begin{example}[Binary distributions, \(d \geq 3\)] \label{Ex.Binary} 
Suppose \(\bX = (X_1, \dots, X_d)\), \(d \geq 3\), is a \(d\)-dimensional binary random vector taking on values in \(\{0,1\}^{d}\).
Denote by \(p_{i_1,\dots,i_d} := \bbP(X_1=i_1, \dots, X_d=i_d)\) the corresponding probability of occurrence, and assume that \(p_{i_1,\dots,i_d} > 0\) for all \((i_1,\dots,i_d) \in \{0,1\}^{d}\). Then 
\begin{align*}
\begin{array}{ccccccc}
\dMTP(\bX) & 
\Longleftrightarrow & \cLOsMTP(\bX) & 
\Longleftrightarrow & \LOSD(\bX) & 
\Longrightarrow & \LOMTP(\bX) 
\end{array}  
\end{align*}
with the reverse implication being invalid. We verify the result in two steps.
\begin{enumerate}[{\rm (1)}]
\item We first show that \(\LOSD(\bX)\) implies \(\dMTP(\bX)\). 
The equivalences then follow by 
Corollary \ref{Cor.MTP2}.
\\
Therefore, assume that \(\LOSD(\bX)\) holds. 
We verify the assertion by applying \cite[Proposition 2.1]{karlin_1980}, which states that a strictly positive function \(f\) is \(\mtp\) if, for every pair of arguments, the corresponding bivariate function obtained by holding all remaining arguments fixed is totally positive of order two.
For convenience, we fix coordinates \(3\) to \(d\) and verify total positivity of order two for the function \((i_1,i_2) \mapsto \bbP(X_1=i_1, X_2=i_2, \bX_3 = {\bf i}_3)\), 
where \(\bX_3 = (X_3, \dots, X_d)\) and \({\bf i}_3 = (i_3, \dots, i_d)\).
Since \(\LOSD(\bX)\) holds, we obtain
\begin{align*}
    \frac{\bbP(X_1=0 \,|\, X_2=1, \bX_3 = {\bf i}_3)}
         {\bbP(X_1=0 \,|\, X_2=0, \bX_3 = {\bf i}_3)} 
    & \leq 1 
    & \frac{\bbP(X_1=1 \,|\, X_2=0, \bX_3 = {\bf i}_3)}
           {\bbP(X_1=1 \,|\, X_2=1, \bX_3 = {\bf i}_3)} 
    & \leq 1
\end{align*}
due to \(\LOSD(X_1 | (X_2,\bX_3))\), and hence
\begin{align*}
  & \frac{\bbP(X_1=0, X_2=1, \bX_3 = {\bf i}_3) \, \bbP(X_1=1, X_2=0, \bX_3 = {\bf i}_3)}
       {\bbP(X_1=0, X_2=0, \bX_3 = {\bf i}_3) \, \bbP(X_1=1, X_2=1, \bX_3 = {\bf i}_3)}
  \\
  & = \frac{\bbP(X_1=0 \,|\, X_2=1, \bX_3 = {\bf i}_3)}           
           {\bbP(X_1=0 \,|\, X_2=0, \bX_3 = {\bf i}_3)} \,        
      \frac{\bbP(X_1=1 \,|\, X_2=0, \bX_3 = {\bf i}_3)}
           {\bbP(X_1=1 \,|\, X_2=1, \bX_3 = {\bf i}_3)} 
  \leq 1\,.
\end{align*}
The same argument applies to any pair of coordinates.
This proves the implication \(\LOSD(\bX) \Longrightarrow \dMTP(\bX)\).

\item 
We now present a counterexample that is \(\MTP(\bX)\) but not  \(\dMTP(\bX)\):
Suppose \(\bX \in \{0,1\}^3\) with probabilities of occurrence
\begin{center}
$\displaystyle \begin{array}{ c|c c c c c c c c }
\bx = (x_{1},x_{2},x_{3}) & ( 0,0,0) & ( 1,0,0) & ( 0,1,0) & ( 0,0,1) & ( 1,1,0) & ( 1,0,1) & ( 0,1,1) & ( 1,1,1)\\
\hline
\bbP(\bX = \bx) & 0.2 & 0.05 & 0.05 & 0.05 & 0.2 & 0.2 & 0.1 & 0.15
\end{array}$
\end{center}
Then it is straightforward to verify \(\MTP(\bX)\), but \(p_{111} p_{100} = 3/400 < 16/400= p_{110} p_{101}\) and hence  \(\dMTP(\bX)\) fails.
\end{enumerate}
\end{example}

For trivariate binary distributions, we can prove a slightly stronger result under weaker assumptions, which in particular permit zero probabilities; see Example \ref{Ex.Binary.d=3}.

\SF{Before turning to the dependence properties of Archimedean copulas, we note that multivariate \(t\)-distributions fail to satisfy \(\cMTP\).}

\begin{example} \label{Ex.t} 
\SF{If \(\bX\) follows a multivariate \(t\)-distribution, then \(\cMTP(\bX)\) fails.  
\\
Indeed, according to \cite[Proposition 4.3]{rossell2021} \(\CI(\bX)\) fails, and the result therefore follows from Theorem \ref{Thm.MTP2.SI} and Corollary \ref{Cor.MTP2}.}
\end{example}

\subsection{Archimedean copulas} \label{Subsect:Archimedean}

In this section, we study continuous random vectors \(\bX\) that are associated with an Archimedean copula and provide a characterization of \(\cMTP(\bX)\) in terms of the Archimedean generator.
In doing so, we extend results in \cite{muller_archimedean_2005} who characterize \(\PLOD(\bX)\), \(\CI(\bX)\), \(\CIS(\bX)\), and \(\dMTP(\bX)\), as well as those in \cite{fuchs_total_2023} who study characterizations of \(\cMTP(\bX)\) in the bivariate case \(d=2\).
In particular, it has been shown in \cite[Theorem 4.6]{fuchs_total_2023} that, for \(d=2\), \(\cMTP(\bX)\) is equivalent to \(\LOSD(\bX)\) 
and, in this specific case, the different notions of positive dependence are linked by the following chain of implications:
$$
\dMTP(\bX) 
 \Longrightarrow  
 \cLOsMTP(\bX) 
 \Longleftrightarrow
 \LOSD(\bX) 
 \Longrightarrow
 \LOMTP(\bX) 
 \Longleftrightarrow
 \LTD(\bX) 
 \Longrightarrow 
 \PLOD(\bX).
$$
Since the connections in the bivariate case are already well understood, we focus on dimensions \(d \geq 3\).
Moreover, by Remark \ref{Rem:Copula}, it suffices to consider the probability integral transform \(\bU = (F_{X_1}(X_1), \dots, F_{X_d}(X_d))\), whose joint distribution function is the copula \(C_\bX\). 

Let \(\psi: [0,\infty) \to [0,1]\) be a continuous, non-increasing function that is strictly decreasing on \([0,t_0]\) with \(t_0:= \inf\{t: \psi(t)=0\}\) and such that \(\psi(0)=1\) and \(\lim_{t \to \infty}\psi(t)=0 =: \psi(\infty)\). 
An Archimedean copula \(C_{\psi}: \bbI^d \to \bbI\), \(d \geq 2\), then is a copula generated by the function \(\psi\), that is, for every \(\bu \in \bbI^d\)
\begin{align*}
  C_{\psi}(\bu) 
  & = \psi \big(\varphi(u_1) + \dots + \varphi(u_d)\big)\,,
\end{align*}
where \(\varphi: \bbI \to  [0,\infty]\) denotes the pseudo-inverse of \(\psi\) given by 
\begin{align*}
  \varphi(t) 
  & = \begin{cases}
      \psi ^{-1}( t), & t\in  (0,1],\\
      t_0,            & t=0.
\end{cases}
\end{align*}
The function \(\varphi\) is continuous, strictly decreasing with \(\varphi(1) = 0\).
\SF{If \(\varphi(0) = \infty\), then \(\psi\) is called strict, and, by convention, so is \(C_{\psi}\); otherwise, both are called non-strict.}
We refer to \(\psi\) as the generator of \(C_{\psi}\).

According to \cite[Theorem 2.2]{mcneil_multivariate_2009}, the copula \(C_{\psi}\) is Archimedean if and only if the generator \(\psi\) is $d$-monotone on \([0,\infty)\), that is, 
\begin{enumerate}[(i)]
    \item it is \((d-2)\)-times differentiable,
    \item the derivatives satisfy 
\begin{align}\label{Def:dMon}
    (-1)^k \, \psi^{(k)}(t) \geq 0\,, \qquad k \in \{0,1,\dots,d-2\},
\end{align}
for every \(t \in (0,\infty)\), and 
    \item \((-1)^{d-2} \psi^{(d-2)}\) is non-increasing and convex on \([0,\infty)\).
\end{enumerate}
If the generator has derivatives of all orders and the inequality in \eqref{Def:dMon} holds for all \(k \in \mathbb{N}\), then \(\psi\) is said to be \emph{completely monotone}.



Before characterizing \(\cLOMTP(\bU)\) for \(\bU\) distributed according to an Archimedean copula in terms of the generator \(\psi\), in Corollary \ref{Cor:DepProp} we first recall known (or by Lemma \ref{lem:logconvex} straightforward) characterizations of the remaining dependence properties; see, e.g.~\cite{nelsen_introduction_2007, fuchs_total_2023, muller_archimedean_2005}. 
Recall that every \(\PLOD(\bU)\) random vector \(\bU\) that is distributed according to the Archimedean copula \(C_\psi\) fulfills \(C_\psi(\bu) \geq \Pi(\bu) > 0\) for all \(\bu \in (0,1]^d\) which implies that \(C_\psi\) is strict, hence \(\psi|_{(0,\infty)} >0\).
In what follows, we may therefore restrict to strict Archimedean copulas. 
A function \(f: [0,\infty) \mapsto (0,\infty)\) is said to be log-convex if \(\log \circ f|_{(0,\infty)}\) is convex.

\begin{corollary}\label{Cor:DepProp}
Let \(\bU\) be distributed according to the strict Archimedean copula \(C_\psi\) with generator \(\psi\) (sufficiently often differentiable).
Then 
\begin{enumerate}[(1)]
    \item\label{Cor:DepProp:1} \(\dMTP(\bU)\) if and only if \((-1)^d\psi^{(d)}\) is log-convex.
    \item\label{Cor:DepProp:2} \(\CIS(\bU)\) if and only if \(\CI(\bU)\) if and only if \((-1)^{d-1} \psi^{(d-1)}\) is log-convex.
    \item\label{Cor:DepProp:3} \(\LOMTP(\bU)\) if and only if \(\LTD(\bU)\) if and only if \(\psi\) is log-convex.
    \item\label{Cor:DepProp:5} \(\PLOD(\bU)\) if and only if \(\log \circ \psi\) is superadditive.
\end{enumerate}
\end{corollary}
\begin{proof}
\eqref{Cor:DepProp:1} and \eqref{Cor:DepProp:2} are due to \cite[Theorem 2.11]{muller_archimedean_2005} and \cite[Theorem 2.8]{muller_archimedean_2005}, \eqref{Cor:DepProp:3} is immediate from Lemma \ref{lem:logconvex}, and 
\eqref{Cor:DepProp:5} is straightforward.
\end{proof}

We now characterize \(\cMTP(\bU)\) and \(\LOSD(\bU)\) and show that in the Archimedean setting both properties are equivalent. 
In addition to strictness of the Archimedean copula, we assume for convenience that \(\psi\) is at least \((d-1)\)--times differentiable. Under this differentiability assumption, the sign condition \eqref{Def:dMon} extends to order \(d-1\).


By \cite[Theorem 1]{kasper_convergence_2024}, (a version of) the regular conditional distribution of the random vector \(\bU = (\bU_A,\bU_B)\), where \(A\) and \(B\) form a partition of \(\{1,\dots,d\}\) with \(|A|=k\) and \(|B|=d-k\), and where \(\bU\) is distributed according to the strict Archimedean copula \(C_\psi\) with generator \(\psi\), is given, for \(\bu_B \in (0,1]^{d-k}\), by 
\begin{equation} \label{MKAC}
  \bbP(\bU_B \leq \bu_B \,|\, \bU_A = \bu_A) 
  = \begin{cases}
  1
  & \textrm{if } \min(\bu_A) = 1\,, \\
  \frac{\psi^{(k)} (\sum_{i \in A} \varphi(u_i) + \sum_{j \in B} \varphi(u_j) )}{\psi^{(k)} (\sum_{i \in A} \varphi(u_i))} 
  & \textrm{if } \min(\bu_A) \in (0,1)\,.
\end{cases}
\end{equation}
Since the marginals of Archimedean copulas are absolutely continuous \cite[Proposition 4.1]{mcneil_multivariate_2009}, the derivatives \(\psi^{(l)}\), \(l \in \{1,\dots,d-1\}\), are absolutely continuous on \((0,\infty)\) and the Lebesgue density \(c_A\) of \(\bU_A\) exists \cite[Proposition 4.2]{mcneil_multivariate_2009} so that, due to disintegration, 
\begin{align*}
    \bbP(\bU_B \leq \bu_B \,|\, \bU_A = \bu_A) \, c_A(\bu_A) 
    & = 
        \frac{\psi^{(k)} (\sum_{i \in A} \varphi(u_i) + \sum_{j \in B} \varphi(u_j) )}{\prod_{i \in A} \psi^{'} (\varphi(u_i))} 
\end{align*}
for all \((\bu_A,\bu_B) \in (0,1)^k \times (0,1]^{d-k}\).

\begin{theorem} \label{Cor:DepProp.ArchC.cMTP2}
Let \(\bU\) be distributed according to the strict Archimedean copula \(C_\psi\) whose generator \(\psi\) is \((d-1)\)--times differentiable. Further, let \(A\) and \(B\) form a partition of \(\{1,\dots,d\}\) with \(|A|=k\) and \(|B|=d-k\).
Then the following statements are equivalent: 
\begin{enumerate}[(a)]
\item \label{Cor:DepProp.ArchC.cMTP2:1} \(\cLOsMTP (\bU_B|\bU_A)\).
\item \label{Cor:DepProp.ArchC.cMTP2:2} \(\LOSD(\bU_B|\bU_A)\).
\item \label{Cor:DepProp.ArchC.cMTP2:3} \((-1)^k \psi^{(k)}\) is log-convex. 
\end{enumerate}
\end{theorem}
\begin{proof}
Due to Theorem \ref{Thm.MTP2.SI}\eqref{Thm.MTP2.SI:1}, \eqref{Cor:DepProp.ArchC.cMTP2:1} implies \eqref{Cor:DepProp.ArchC.cMTP2:2}.
Now, assume that \eqref{Cor:DepProp.ArchC.cMTP2:2} holds.
Then, by definition, the mapping \((0,1)^k \mapsto [0,1]\) given by
\begin{align*}
    \bu_A
    & \mapsto \bbP(\bU_B \leq \bu_B \,|\, \bU_A = \bu_A)
      = \frac{\psi^{(k)} (\sum_{i \in A} \varphi(u_i) + \sum_{j \in B} \varphi(u_j) )}{\psi^{(k)} (\sum_{i \in A} \varphi(u_i))}
      \text{ is non-increasing for all } \bu_B \in (0,1]^{d-k}.
\end{align*}
Log-convexity of \((-1)^k \psi^{(k)}\) now follows from Lemma \ref{lem:logconvex}.
\\
Finally, assume that \eqref{Cor:DepProp.ArchC.cMTP2:3} holds, 
and set \(h := \log \circ ((-1)^k \psi^{(k)})\).
Then, for \(\bu_A, \bv_A \in (0,1)^k\) and \(\bu_B, \bv_B \in (0,1]^{d-k}\), we obtain \((-1)^k \psi^{(k)} \big(\sum_{i \in A}  \varphi(u_i \wedge v_i) + \sum_{j \in B} \varphi(u_j \wedge v_j) \big) > 0\) and
\begin{align} \label{Cor:DepProp.ArchC.cMTP2:P1}
    & \log \left( 
    \frac{\bbP(\bU_B \leq \bu_B \wedge \bv_B \,|\, \bU_A = \bu_A \wedge \bv_A) \, c_A(\bu_A \wedge \bv_A) \, 
          \bbP(\bU_B \leq \bu_B \vee \bv_B \,|\, \bU_A = \bu_A \vee \bv_A) \, c_A(\bu_A \vee \bv_A)}
         {\bbP(\bU_B \leq \bu_B \,|\, \bU_A = \bu_A) \, c_A(\bu_A) \, \bbP(\bU_B \leq \bv_B \,|\, \bU_A = \bv_A) \, c_A(\bv_A)} \right) \notag
    \\
    &   =  \log \left( 
    \frac{\psi^{(k)} (\sum_{i \in A} \varphi(u_i \wedge v_i) + \sum_{j \in B} \varphi(u_j \wedge v_j) )}
         {\psi^{(k)} (\sum_{i \in A} \varphi(u_i) + \sum_{j \in B} \varphi(u_j))} 
    \frac{\psi^{(k)} (\sum_{i \in A} \varphi(u_i \vee v_i) + \sum_{j \in B} \varphi(u_j \vee v_j) )}
         {\psi^{(k)} (\sum_{i \in A} \varphi(v_i) + \sum_{j \in B} \varphi(v_j))} \right) \notag
    \\
    & \qquad + \log \left( 
    \frac{\prod_{i \in A} \psi^{'} (\varphi(u_i))} 
         {\prod_{i \in A} \psi^{'} (\varphi(u_i \wedge v_i))} 
    \frac{\prod_{i \in A} \psi^{'} (\varphi(v_i))} 
         {\prod_{i \in A} \psi^{'} (\varphi(u_i \vee v_i))} 
    \right) \notag
    \\
    &   =  \log \left( 
    \frac{\psi^{(k)} (\sum_{i \in A} \varphi(u_i \wedge v_i) + \sum_{j \in B} \varphi(u_j \wedge v_j) )}
         {\psi^{(k)} (\sum_{i \in A} \varphi(u_i) + \sum_{j \in B} \varphi(u_j))} 
    \frac{\psi^{(k)} (\sum_{i \in A} \varphi(u_i \vee v_i) + \sum_{j \in B} \varphi(u_j \vee v_j) )}
         {\psi^{(k)} (\sum_{i \in A} \varphi(v_i) + \sum_{j \in B} \varphi(v_j))} \right) 
         + \log(1)
         \notag
    \\
    &   =  
    h \left( \underbrace{\sum_{i \in A} \varphi(u_i \wedge v_i) + \sum_{j \in B} \varphi(u_j \wedge v_j) }_{=:b} \right) + 
    h \left( \underbrace{\sum_{i \in A} \varphi(u_i \vee v_i) + \sum_{j \in B} \varphi(u_j \vee v_j)}_{=:a} \right) \notag
    \\
    & \quad  
    -
    h \left( \underbrace{\sum_{i \in A} \varphi(u_i) + \sum_{j \in B} \varphi(u_j)}_{=:c} \right) - 
    h \left( \underbrace{\sum_{i \in A} \varphi(v_i) + \sum_{j \in B} \varphi(v_j)}_{=:d} \right)\,. 
\end{align}
Since \(h\) is convex,
\( a \leq \min \left\{c,d\right\} \leq \max \left\{c,d\right\} \leq b\), and 
\(a+b=c+d\), it follows from \cite[Corollary 1.4.3]{niculescu2004} that \eqref{Cor:DepProp.ArchC.cMTP2:P1} $\geq 0$. 
The boundary cases with \((\bu_A,\bu_B) \in [0,1]^k \backslash (0,1)^{k} \times [0,1]^{d-k} \backslash (0,1]^{d-k}\) follow by choosing appropriate versions of the regular conditional distribution and of the density. Recall that this set has \(\bbP^{\bU_A}\) measure zero.
Thus, \(\cLOsMTP (\bU_B|\bU_A)\) follows.
\end{proof}

The following result concludes the discussion and is an immediate consequence of Theorem \ref{Cor:DepProp.ArchC.cMTP2}, Corollary \ref{Cor:DepProp}, and the fact that log-convexity of \((-1)^{d-1} \psi^{(d-1)}\) implies that of \((-1)^{k} \psi^{(k)}\) for all \(k \in \{1,\dots,d-1\}\) \cite[Theorem 2.8]{muller_archimedean_2005}.
In the Archimedean setting considered here \(\LOSD(\bX) \Longleftrightarrow \CIS(\bX) \Longleftrightarrow \CI(\bX)\).

\begin{corollary}\label{Cor:DepProp.ArchC.SI}
Let \(\bU\) be distributed according to the strict Archimedean copula \(C_\psi\) whose generator \(\psi\) is \((d-1)\)--times differentiable.
Then the following statements are equivalent: 
\begin{enumerate}[(a)]
\item\label{Cor:DepProp.ArchC.SI:1} \(\cMTP(\bU)\).
\item\label{Cor:DepProp.ArchC.SI:2} \(\LOSD(\bU)\). 
\item\label{Cor:DepProp.ArchC.SI:3} \((-1)^{d-1} \psi^{(d-1)}\) is log-convex.
\end{enumerate}
\end{corollary}

Fig.~\ref{Fig:DepProp.Arch.2} summarizes the findings of Subsection \ref{Subsect:Archimedean} for continuous random vectors \(\bX\) that are associated with a strict Archimedean copula. 
The stated implications and equivalences also hold in the bivariate case.
\SF{The class of Archimedean copulas with completely monotone generators is rich and includes the Clayton and Frank copulas with positive dependence parameters, as well as the Gumbel-Hougaard and Joe copulas, among many others; we refer to \cite[Chapter 4]{nelsen_introduction_2007} for further examples. As a consequence, these copulas are all \(\dMTP\) and hence \(\cMTP\).
By contrast, examples of positively dependent Archimedean copulas whose generators are not completely monotone appear to be relatively rare. 
Notable examples are provided in Example \ref{Ex.Arch.Counter} below and \cite[Remark 2.13]{muller_archimedean_2005}; these examples show that the reverse implications in Fig.~\ref{Fig:DepProp.Arch.2}, except for those explicitly indicated as equivalences, do not hold in general.}

\begin{example} \label{Ex.Arch.Counter}
\SF{For \(\alpha \in (0,\infty)\), the mapping \(\psi_\alpha: [0,\infty) \to [0,1]\) defined by
\begin{align*}
    \psi_\alpha(t)
    & := \frac{(1-t-\alpha) + \sqrt{(t+\alpha-1)^2 + 4\alpha}}{2}
\end{align*}
is an Archimedean generator, and the corresponding bivariate Archimedean copula \(C_\alpha\) coincides with family 16 in \cite[Table 4.2]{nelsen_introduction_2007}.
Then, \((\bU = U_1,U_2)\) distributed according to \(C_\alpha\) satisfies
\begin{itemize}
    \item \(\dMTP(\bU)\) if and only if \(\alpha \geq 3+2\sqrt{2}\),
    \item \(\cMTP(\bU)\) if and only if \(\alpha \geq 3\), and
    \item \(\MTP(\bU)\) if and only if \(\alpha \geq 1\).
\end{itemize}}
\end{example}

\vspace{-3mm}
\begin{figure}[ht]
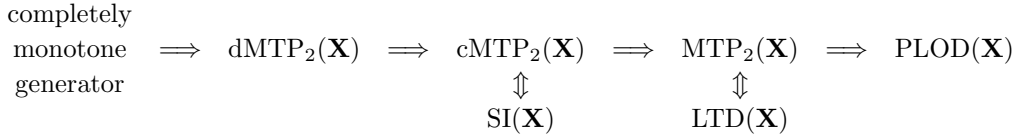
 
\begin{center}
$$
\begin{array}{ccccccccc}
 \textrm{completely} &&&&&&&&
 \\
 \textrm{monotone}  
 & \Longrightarrow & \dMTP(\bX) 
 & \Longrightarrow & \cLOsMTP(\bX)
 & \Longrightarrow & \LOMTP(\bX) 
 & \Longrightarrow & \PLOD(\bX)
 \\
 \textrm{generator} 
 &&&& \Updownarrow && \Updownarrow &&
 \\
 &&&& \LOSD(\bX) && \LTD(\bX) && 
\end{array}
$$
\end{center}
\caption{Interrelations among the different notions of positive dependence for a continuous random vector \(\bX\) whose copula is Archimedean. Arrows pointing from left to right indicate strict implications; that is, the property on the left is strictly stronger than, and therefore implies, the property on the right.}
\label{Fig:DepProp.Arch.2}
\end{figure}

\section{Application to Markov products and measures of directed dependence}
\label{Sec:markovproduct}

We now investigate stability results for Markov product transformations and, consequently, stability results within the framework of conditionally independent random vectors. We then highlight the implications of these results for the comparison of measures of directed dependence.

\subsection{Stability results for Markov product transformations}

Let \(\bX\) and \(\bY\) be 
random vectors of dimension \(p \geq 1\) and \(q \geq 1\), respectively, and let \(Z\) be a 
random variable. 
For the distribution functions \(F\) and \(G\) corresponding to \((\bX,Z)\) and \((\bY,Z)\), respectively, define the functions
\begin{align}
    F \star G \, (\bx,\by,z)
    & := \int_{(-\infty,z]} \bbP(\bX \leq \bx \,|\, Z = t) \, \bbP(\bY \leq \by \,|\, Z = t) \de \bbP^Z(t)
    \label{Def:MP.DepProp} \\
    F \ast G \, (\bx,\by)
    & := \lim_{z \to \infty} F \star G \, (\bx,\by,z). \label{Def:MP.DepProp2}
\end{align}
Both \(F \star G\) and \(F \ast G\) are distribution functions; while \(F \star G\) is \((p+q+1)\)--dimensional, \(F \ast G\) is \((p+q)\)--dimensional.
The distribution functions \(F \star G\) and \(F \ast G\) represent conditionally independent compositions of the random vectors \(\bX\) and \(\bY\) in the following sense: \((\bX,\bY,Z)\) has distribution function \(F \star G\) if and only if \(\bX\) and \(\bY\) are conditionally independent given \(Z\); cf. \cite[Lemma 3]{overbeck_multivariate_2015}.
In this case \((\bX,\bY)\) is distributed according to \(F \ast G\), i.e.~\((\bX,\bY) \sim F \ast G\).
Both distribution functions are suitable for modeling Markov processes \cite{fallat_total_2017,durante_principles_2015,overbeck_multivariate_2015}.
%
%

The first result connects the positive dependence properties examined in Section \ref{Sec:MTP2} of \((\bX,Z)\) and \((\bY,Z)\) with the corresponding property of the combined random vector \((\bX,\bY,Z)\), whose distribution function is \(F \star G\).
Assertion \eqref{Thm:MP.DepProp:1} of Theorem \ref{Thm:MP.DepProp} is taken from \cite[Proposition 7.1]{fallat_total_2017}.

\begin{theorem}\label{Thm:MP.DepProp}
Consider the random vector \((\bX,\bY,Z)\) such that \(\bX\) and \(\bY\) are conditionally independent given \(Z\).
Then
\begin{enumerate}[(1)]
\item\label{Thm:MP.DepProp:1} 
\(\dMTP(\bX,Z)\) and \(\dMTP(\bY,Z)\) if and only if \(\dMTP(\bX,\bY,Z)\).
\item\label{Thm:MP.DepProp:2} 
\(\cMTP(\bX|Z)\) and \(\cMTP(\bY|Z)\) if and only if \(\cMTP((\bX,\bY)|Z)\).
\item\label{Thm:MP.DepProp:3} 
\(\LOSD(\bX|Z)\) and \(\LOSD(\bY|Z)\) if and only if \(\LOSD((\bX,\bY)|Z)\).
\end{enumerate}
\end{theorem}
\begin{proof}
It remains to prove \eqref{Thm:MP.DepProp:2} and \eqref{Thm:MP.DepProp:3}.
Therefore, assume \(\cMTP(\bX|Z)\) and \(\cMTP(\bY|Z)\) and define 
\(f(\bx,z) := \bbP(\bX \leq \bx \,|\, Z = z)\) and
\(g(\by,z) := \bbP(\bY \leq \by \,|\, Z = z)\).
Then, by \eqref{Def:cMTP2:Eq}, \(f\) and \(g\) are \(\mtp\).
\(\cMTP((\bX,\bY)|Z)\) now follows from \eqref{Def:cMTP2:Eq} and the fact that the conditional distribution function of \((\bX,\bY)\) given \(Z\) satisfies
\begin{align}\label{Thm:MP.DepProp:Eq1}
  \bbP((\bX,\bY) \leq (\bx,\by) \,|\, Z = z)   
   & = \bbP(\bX \leq \bx \,|\, Z = z) \, \bbP(\bY \leq \by \,|\, Z = z)
     = f(\bx,z) \, g(\by,z)
\end{align}
for \(\bbP^Z\)-almost all \(z \in \bbR\) and all \(\bx,\by\). The latter product \(fg\), however, is \(\mtp\) due to \cite[Proposition 3.3]{karlin_1980}.
Since \(\bX\) and \(\bY\) are the marginals of \((\bX,\bY)\), the reverse direction in \eqref{Thm:MP.DepProp:2} follows from Corollary \ref{Cor.cMTP2:Margins}.
The forward direction in \eqref{Thm:MP.DepProp:3} is due to \eqref{Thm:MP.DepProp:Eq1}, and the reverse direction follows from Remark \ref{Cor.SI:Margins}\eqref{Cor.SI:Margins:2}.
\end{proof}

Theorem \ref{Thm:MP.DepProp}\eqref{Thm:MP.DepProp:1} is specific to the case where $Z$ is one-dimensional. 
It does not extend to vector-valued $Z$; see \cite[Example 7.2]{fallat_total_2017}.
The next example shows that the same limitation applies to Theorem \ref{Thm:MP.DepProp}\eqref{Thm:MP.DepProp:2} as well. 
In particular, Example \ref{Ex:MP.DepProp} demonstrates that imposing additional dependence assumptions on \(Z\) does not remedy the failure of the result in the vector-valued setting.
By contrast, Theorem \ref{Thm:MP.DepProp}\eqref{Thm:MP.DepProp:3} remains valid for vector-valued \(Z\); this is immediate from the vector-valued analogue of \eqref{Thm:MP.DepProp:Eq1} together with Remark \ref{Cor.SI:Margins}\eqref{Cor.SI:Margins:2}.

\begin{example} \label{Ex:MP.DepProp}
Consider the binary random vector \(\bX = (X_1,\dots,X_4) \in \{0,1\}^4\) with the following strictly positive probabilities:
\begin{align*}
  (p_{0000},p_{1000},p_{0100},p_{1100},p_{0001},p_{1001},p_{0101},p_{1101})
  & = (1080,270,120,30,980,420,420,180)/6500,
  \\
  (p_{0010},p_{1010},p_{0110},p_{1110},p_{0011},p_{1011},p_{0111},p_{1111})
  & = (480,320,120,80,400,600,400,600)/6500.
\end{align*}
Then \(X_1\) and \(X_2\) are conditionally independent given \((X_3,X_4)\), and the joint distribution of \((X_1,X_2,X_3,X_4)\) is obtained as the Markov product of the distributions of \((X_1,X_3,X_4)\) and \((X_2,X_3,X_4)\), using the vector-valued analogue of \eqref{Def:MP.DepProp}.
A direct calculation shows that both \(\cMTP(X_1|(X_3,X_4))\) and \(\cMTP(X_2|(X_3,X_4))\) hold, however,
\begin{align*}
    & \bbP((X_1,X_2) \leq (0,0) \,|\, (X_3,X_4) = (0,1)) \, f_{(X_3,X_4)} (0,1)
    \; \bbP((X_1,X_2) \leq (0,0) \,|\, (X_3,X_4) = (1,0)) \, f_{(X_3,X_4)} (1,0)
    \\
    & = \frac{1183}{6250} > \frac{1521}{8750}
    \\
    & = \bbP((X_1,X_2) \leq (0,0) \,|\, (X_3,X_4) = (0,0)) \, f_{(X_3,X_4)} (0,0)
    \; \bbP((X_1,X_2) \leq (0,0) \,|\, (X_3,X_4) = (1,1)) \, f_{(X_3,X_4)} (1,1)\,,
\end{align*}
hence \(\cMTP((X_1,X_2)|(X_3,X_4))\) fails, despite the fact that $\dMTP(X_3,X_4)$ holds.
\end{example}

As a consequence of Theorem \ref{Thm:MP.DepProp}, we now relate the dependence properties of \((\bX,Z)\) and \((\bY,Z)\) to that of the random vector \((\bX,\bY)\).

\begin{corollary}\label{Cor:MP.DepProp}
Consider the random vector \((\bX,\bY,Z)\) such that \(\bX\) and \(\bY\) are conditionally independent given \(Z\).
Then
\begin{enumerate}[(1)]
\item\label{Cor:MP.DepProp:1} \(\dMTP(\bX,Z)\) and \(\dMTP(\bY,Z)\) imply \(\dMTP(\bX,\bY)\).
\item\label{Cor:MP.DepProp:2} \(\cMTP(\bX|Z)\) and \(\cMTP(\bY|Z)\) imply \(\MTP(\bX,\bY)\).\pagebreak
\end{enumerate}
\end{corollary}
\begin{proof}
Assertion \eqref{Cor:MP.DepProp:1} is immediate from Theorem \ref{Thm:MP.DepProp}\eqref{Thm:MP.DepProp:1} and Remark \ref{ClosedMargins.2}.
Assertion \eqref{Cor:MP.DepProp:2} follows from Theorem \ref{Thm:MP.DepProp}\eqref{Thm:MP.DepProp:2}, Theorem \ref{Thm.MTP2}\eqref{Thm.MTP2:1} and Remark \ref{ClosedMargins.2}.
\end{proof}

For the case \(p=1=q\), Corollary \ref{Cor:MP.DepProp} can be improved as follows.

\begin{corollary}\label{Cor:MP.DepProp.d2}
Consider the random vector \((X,Y,Z)\) such that \(X\) and \(Y\) are conditionally independent given \(Z\).
Then
\begin{enumerate}[(1)]
\item\label{Cor:MP.DepProp.d2:1} \(\dMTP(X,Z)\) and \(\dMTP(Y,Z)\) imply \(\dMTP(X,Y)\).
\item\label{Cor:MP.DepProp.d2:2} \(\cMTP(X|Z)\) and \(\cMTP(Y|Z)\) imply \(\MTP(X,Y)\).
\item\label{Cor:MP.DepProp.d2:3} \(\LOSD(Z|X)\) and \(\LOSD(Y|Z)\) imply \(\LOSD(Y|X)\).
\item\label{Cor:MP.DepProp.d2:4} \(\LOSD(X|Z)\) and \(\LOSD(Z|Y)\) imply \(\LOSD(X|Y)\).
\item\label{Cor:MP.DepProp.d2:5} \(\LOSD(X,Z)\) and \(\LOSD(Y,Z)\) imply \(\LOSD(X,Y)\).
\end{enumerate}
\end{corollary}
\begin{proof}
\eqref{Cor:MP.DepProp.d2:1} and \eqref{Cor:MP.DepProp.d2:2} are special cases of Corollary \ref{Cor:MP.DepProp}, and \eqref{Cor:MP.DepProp.d2:3}, \eqref{Cor:MP.DepProp.d2:4} and \eqref{Cor:MP.DepProp.d2:5} are due to \cite[Theorem 3.2]{joe1994}; see also \cite[Corollary 4.1]{Siburg-2021} and \cite[Lemma B.2]{ansari2025}.
\end{proof}

Corollary \ref{Cor:MP.DepProp.d2} has important consequences, which we examine in the next subsection.

\subsection{Comparison results for measures of directed dependence}

Let us now apply Corollary \ref{Cor:MP.DepProp.d2} to a variant of the Markov product recently studied in the context of measures of directed dependence. 
Consider, therefore, a random vector \((Y,X)\) with non-degenerate \(Y\) and denote by \(Y'\) a conditionally independent copy of \(Y\) given \(X\,,\) i.e.
\begin{align} \label{MP:Def}
  (Y'\,|\,X=x) \eqd (Y\,|\,X=x) 
  \textrm{ for } \bbP^X\text{-almost all } x \in \bbR \text{ and } Y \perp Y' \mid X \,,
\end{align}
where \(\eqd\) indicates equality in distribution and \(Y \perp Y' \mid X\) denotes conditional independence of \(Y\) and \(Y'\) given \(X\). 
Then, according to \cite{fuchs2024JMVA,ansari2025Cont}, if \((Y,X)\) has distribution function \(F\), then \((Y,Y')\) has distribution function \(F \ast F\) as given in \eqref{Def:MP.DepProp2}.

The transformed random vector \((Y,Y')\) associated with \((Y,X)\) underlies several measures of directed dependence that quantify the influence that \(X\) has on \(Y\).
These include the recently introduced Chatterjee's rank correlation \(\xi\) \cite{chatterjee2020, siburg2013}
given by 
\begin{align*}
  \xi(Y,X) 
  & := \frac{\int_{\bbR} \operatorname{Var}(P(Y \geq y \mid X)) \, \de P^Y(y)}
            {\int_{\bbR} \operatorname{Var}(\mathds{1}_{\{Y \geq y\}}) \, \de P^Y(y)}
     =  \frac{\int_{\bbR} P(Y \geq y, Y' \geq y) - \bbP(Y \geq y)^2 \, \de P^Y(y)}
            {\int_{\bbR} \bbP(Y \geq y) - \bbP(Y \geq y)^2 \, \de P^Y(y)}\,,
\end{align*}
and the normalized Pearson's correlation ratio \(R^2\), introduced in \cite{pearson1905} and given by 
\begin{align*}
  R^2(Y,X) 
  & := \frac{\operatorname{Var}(\mathbb{E}(Y|X))}{\operatorname{Var}(Y)}
     = \rho_P(Y,Y')\,,
\end{align*}
where \(\rho_P\) denotes Pearson's correlation. \(R^2\) is also known as Sobol' index \cite{Sobol1993}.
While Chatterjee's \(\xi\) quantifies the scale-invariant extent of functional dependence of \(Y\) on \(X\), Pearson's correlation ratio \(R^2\) measures the fraction of variance explained by the regression function \(h(x) := \mathbb{E}(Y|X=x)\).

A fundamental property of \((Y,Y')\) is that
\begin{itemize}
\item \(Y\) and \(X\) are independent if and only if $Y$ and $Y'$ are independent,
\item \(\operatorname{Var}(\mathbb{E}(Y|X)) = 0\) if and only if $Y$ and $Y'$ are uncorrelated, and 
\item \(Y\) perfectly depends on \(X\), i.e.~$Y = f(X)$ almost surely for some measurable function \(f\), if and only if $Y = Y'$ almost surely;
\end{itemize}
we refer to \cite{fuchs2026MP} for a more detailed study of how dependence structures are transferred under this transformation.\pagebreak
Thus, \((Y,Y')\) preserves the key information about the strength of functional dependence of \(Y\) on \(X\).
Likewise, \(\xi(Y,X) = 0\) if and only if \(Y\) and \(X\) are independent,  \(R^2(Y,X) = 0\) if and only if \(\operatorname{Var}(\mathbb{E}(Y|X)) = 0\), and \(\xi(Y,X) = 1\) if and only if \(R^2(Y,X) = 1\) if and only if $Y$ perfectly depends on $X$.

Only in very few cases is it known how the transformation \(F \mapsto F \ast F\) affects certain distributions: For instance, if \((Y,X)\) is distributed according to a Gaussian distribution with correlation parameter \(\rho\), then \((Y,Y')\) is Gaussian as well with correlation parameter \(\rho^2\) \cite[Example 1]{fuchs2024JMVA}.
On the other hand, if \((Y,X)\) is lower semilinear, then \((Y,Y')\) is again lower semilinear \cite{maislinger2025, fuchs2025IJAR}.
By contrast, to the best of our knowledge, no closed-form expression for \(F \ast F\) is known when \((Y,X)\) follows an Archimedean copula or an extreme-value copula; we refer to \cite{durante_principles_2015} for more information on this class of copulas. 
Moreover, neither class is closed under the Markov product transformation. 
Nevertheless, Corollary \ref{Cor:MP.DepProp.d2}, together with the results of Subsection \ref{Subsect:Archimedean} and the fact that extreme-value copulas are stochastically increasing in both directions \cite{guillem2000}, yields stability results for their dependence properties.

Although \((F \ast F)(y,y) = P(Y \leq y, Y' \leq y) \geq P(Y \leq y) \, \bbP(Y' \leq y)\) for all \(y \in \bbR\), the random vector \((Y,Y')\) does not, in general, satisfy \(\PLOD(Y,Y')\) \cite[Example 2]{fuchs2024JMVA}. 
Therefore, it is important to better understand the stability of dependence properties under this specific Markov product.
According to Corollary \ref{Cor:MP.DepProp.d2}, we have that
\begin{itemize}
\item
\(\dMTP(Y,X)\) implies \(\dMTP(Y,Y')\).
\item
\(\cMTP(Y|X)\) implies \(\MTP(Y,Y')\).
\item
\(\LOSD(Y,X)\) implies \(\LOSD(Y,Y')\).
\end{itemize}
Interestingly, the sufficient condition for \(\LOSD(Y,Y')\) involves both directions \(\SI(Y|X)\) and \(\SI(X|Y)\), while \(\MTP(Y,Y')\) follows from the one-sided condition \(\cMTP(Y|X)\).
In particular, Corollary \ref{Cor:MP.DepProp.d2} implies that the Markov product associated with a random vector \((Y,X)\) whose copula is extreme-value satisfies \(\LOSD(Y,Y')\). It also shows that, if the copula of \((Y,X)\) is Archimedean with Archimedean generator \(\psi\), then the Markov product satisfies
\begin{itemize}
    \item \(\dMTP(Y,Y')\) whenever \(\psi{''}\) is log-convex,
    \item \(\MTP(Y,Y')\) and \(\LOSD(Y,Y')\) whenever \((-\psi)'\) is log-convex.
\end{itemize}
These conclusions hold without identifying the copula family to which the copula of \((Y,Y')\) belongs.

For a continuous random vector \((Y,X)\) whose probability integral transform 
\((F_Y(Y),F_X(X))\) is distributed according to copula \(C\) (interpreted as a distribution function with uniform margins), both Chatterjee's \(\xi\) and (the distribution-free) Pearson's correlation ratio \(R^2\) admit representations in terms of classical measures of concordance of the copula 
\(C \ast C\), namely Spearman's footrule \(\phi\) and Spearman's rho \(\rho_S\), respectively. More precisely \cite{fuchs2024JMVA},
\begin{align}\label{Rep:XiR2}
    \xi(Y,X) & = \xi(F_Y(Y),F_X(X)) = \phi(C \ast C)
    \qquad \textrm{ and } \qquad 
    R^2(F_Y(Y),F_X(X)) = \rho_S(C \ast C)\,.
\end{align}
We now show how the dependence properties inherent in \(C\) yield results concerning the relationship between the dependence measures \(\xi\) and \(R^2\).
To this end, we require the following positive dependence property, together with a general result comparing Spearman's footrule and Spearman's rho:
\SF{As in \cite{nelsen_introduction_2007}, the random vector \((Z_1,Z_2)\) is said to be \(\RTI(Z_2|Z_1)\),} \emph{right tail increasing}, if the mapping (whenever it exists)
\begin{equation*}
  z_1 \mapsto \bbP(Z_2 > z_2 \,|\, Z_1 > z_1) \text{ is non-decreasing for all } z_2 \in \bbR\,.
\end{equation*}
According to \cite{fuchs_total_2023,nelsen_introduction_2007},
\(\RTI(Z_2|Z_1)\) is implied by each of \(\MTP(Z_1,Z_2)\) and \(\LOSD(Z_2|Z_1)\).

\begin{theorem} \label{Thm:FootruleRho}
Let \((Z_1,Z_2)\) be a continuous random vector with copula \(C\).
If \(\LTD(Z_2|Z_1)\) and \(\RTI(Z_2|Z_1)\), then Spearman's rho is at least as large as Spearman's footrule, i.e., \(\rho_S(C) \geq \phi(C)\).
\end{theorem}
\begin{proof}
For \(u,v \in [0,1]\), define \(h_{v}(u) := C(u,v) - uv\), and recall that Spearman's rho and Spearman's footrule \cite{nelsen_introduction_2007, genest_spearmans_2010} are given by
\begin{align*}
  \rho_S(C) 
  & = 12 \int_{(0,1)^2} h_{v}(u) \de \lambda^2(u,v)
  \qquad \textrm{and} \qquad 
  \phi(C) = 6 \int_{(0,1)} h_{v}(v) \de \lambda(v)\,,
\end{align*}
so that 
\begin{align*} 
  \rho_S(C) - \phi(C) 
  & = 6 \int_{(0,1)} 2 \, \left( \int_{(0,1)} h_{v}(u) \de \lambda(u) \right) - h_{v}(v) \de \lambda(v)\,,
\end{align*}
where \(\lambda\) denotes the Lebesgue measure on \((0,1)\).
Now, because \(\LTD(Z_2|Z_1)\) implies that the mapping 
\begin{align*}
  u & \mapsto \frac{h_{v}(u)}{u} = \frac{C(u,v)}{u} - v
\end{align*}
is non-increasing for all \(v \in \bbI\), and \(\RTI(Z_2|Z_1)\) implies that the mapping 
\begin{align*}
  u & \mapsto \frac{h_{v}(u)}{1-u} 
      = \frac{C(u,v)-v}{1-u} + \frac{v-uv}{1-u}
      = \frac{C(u,v)-v}{1-u} + v 
\end{align*}
is non-decreasing for all \(v \in \bbI\), we obtain 
\begin{align*}
  \frac{u}{v} \, h_{v}(v) 
  & \leq h_{v}(u) \qquad \textrm{ for all } 0 < u \leq v < 1,   
    \qquad \textrm{ and }  
  \\
  \frac{1-u}{1-v} \, h_{v}(v) 
  & \leq h_{v}(u) \qquad \textrm{ for all } 0 < v \leq u < 1. 
\end{align*}
This implies 
\begin{align*}
  \int_{(0,1)} h_{v}(u) \de \lambda(u) 
  & \geq \int_{(0,v]} \frac{u}{v} \; h_{v}(v) \de \lambda(u) + \int_{(v,1)} \frac{1-u}{1-v} \; h_{v}(v) \de \lambda(u) 
      =  \frac{1}{2} h_{v}(v)
\end{align*}
for all \(v \in (0,1)\), and eventually 
\(\rho_S(C) - \phi(C) \geq 0\).
\end{proof}

Analogously to Theorem \ref{Thm:FootruleRho}, \cite{fredricks_relationship_2007} proved that the positive dependence properties \(\LTD(Z_2|Z_1)\) and \(\RTI(Z_2|Z_1)\) also imply that Spearman's rho is at least as large as Kendall's tau.

According to \cite{fuchs_total_2023}, each of \(\LTD(Z_2|Z_1)\) and \(\RTI(Z_2|Z_1)\) implies \(\PLOD(Z_1,Z_2)\). The following example demonstrates that \(\PLOD(Z_1,Z_2)\) alone is not strong enough to imply the inequality in Theorem \ref{Thm:FootruleRho}.

\begin{example}
Let the copula of \((Z_1,Z_2)\) be the lower semilinear copula \(S_\delta\) given by
\begin{align*}
  S_{\delta}(u,v) 
  & = \begin{cases}
    u \frac{\delta(v)}{v}, & u < v, \\
    v \frac{\delta(u)}{u}, & u \geq v,
\end{cases}
\quad \textrm{where} \quad 
  \delta(t) := \begin{cases}
    at, & t \leq a,\\
    t^2, & t > a,
\end{cases}
\end{align*} 
for some \(a \in (0,1)\) \cite{durante2008}.
Evidently, \(\PLOD(Z_1,Z_2)\) holds, while  \(\rho_S(S_{\delta}) = a^4 < a^3 =\phi(S_{\delta})\) according to \cite[Example 2.1]{fuchs2025IJAR}. 
\end{example}

Theorem \ref{Thm:FootruleRho} enables a comparison between Chatterjee's rank correlation \(\xi(Y,X)\) and the distribution-free Pearson's correlation ratio \(R^2(F_Y(Y),F_X(X))\), expressed not only through dependence properties of \((Y,Y')\), but also through those of the original random vector \((Y,X)\).

\begin{corollary}\label{Cor:XiR2}
Let \((Y,X)\) be a continuous random vector and \(Y'\) be a conditionally independent copy of \(Y\) given \(X\) as defined in \eqref{MP:Def}. If \(\LTD(Y'|Y)\) and \(\RTI(Y'|Y)\), then
\begin{align} \label{Ineq:XiR2}
    R^2(F_Y(Y),F_X(X))
    & \geq \xi(F_Y(Y),F_X(X)) = \xi(Y,X)\,.
\end{align}
In particular, \eqref{Ineq:XiR2} holds if either \(\cLOMTP(Y|X)\) or \(\LOSD(Y,X)\).
\end{corollary}

Note that \(\cLOMTP(Y|X)\) and \(\LOSD(Y,X)\) are not generally related since \(\cLOMTP(Y|X)\) implies \(\SI(Y|X)\) but it does not generally imply \(\SI(X|Y)\). For example, the random vector \((Y,X)\) with copula 
\begin{align}\label{Ex:XiR2}
  C(v,u) & := vu + v(1-v)^2\,u(1-u)
\end{align}
satisfies \(\cLOMTP(Y|X)\), and hence \(\SI(Y|X)\), but does not satisfy \(\SI(X|Y)\).
Consequently, Corollary \ref{Cor:XiR2} implies that the inequality 
\begin{align*}
  R^2(F_Y(Y),F_X(X)) & \geq \xi(F_Y(Y),F_X(X)) = \xi(Y,X)
\end{align*}
holds for every continuous pair of random variables \((Y,X)\) whose copula is \(C\) in \eqref{Ex:XiR2}, extreme-value, or Archimedean with generator \(\psi\) such that \(-\psi'\) is log-convex.

\section{Discussion}
\label{Sec:Dis}

\textit{Negative dependence concepts:}
\SF{Reversing inequality \eqref{MTP2} defining \(\mtp\) gives rise to the notion of \emph{multivariate reverse rule of order 2} (\(\mrr\)), introduced in \cite{karlinb_1980}.
Notable examples of \(\mrr\) functions include densities (probability mass functions) of the multinomial, multivariate hypergeometric, and Dirichlet distributions.
This notion motivates a negative-dependence analogue of \(\cLOMTP\). 
While the \(\mtp\) property is preserved under composition \cite[Proposition 3.4]{karlin_1980}, a fact central to our analysis as it underlies Lemma \ref{MTP:Help1} and related results in Section \ref{Sec:MTP2}, an analogous result for \(\mrr\) functions does not hold in general. 
Developing a meaningful negative-dependence analogue of \(\cLOMTP\) may therefore require a suitably strengthened version of \(\mrr\). We leave the formulation of such a notion and the study of its potential stability under Markov products for future research.}

\textit{Density-free dependence concepts:}
\SF{The classical definition of \(\dMTP\) requires the underlying distribution to admit a density with respect to a product measure (Remark \ref{Rem.cMTP2:Ex}).
However, this requirement is generally not satisfied by comonotonic random vectors, despite comonotonicity being regarded as the strongest possible form of positive dependence under virtually any reasonable notion of positive dependence. 
\cite{muller2006} therefore introduce a density-free extension of \(\dMTP\) for probability measures in terms of lattice inequalities for probabilities of Borel sets. 
Their definition agrees with the density-based one whenever a density with respect to a product measure exists \cite[Theorem 3.10.14]{muller2002} and has the advantage of being closed under weak convergence \cite[Theorem 2]{muller2006}.
A possible analogue for \(\cLOMTP\) can be defined as follows: 
For a partition \(A, B\) of \(\{1,2,\dots,d\}\), the random vector \(\bX\) is said to be \(\cLOMTPs(\bX_B|\bX_A)\) if 
\begin{align}\label{Def:cMTP2AB:Eq2}
   & \bbP( \bX_A \in C, \bX_B \in (-\infty, \bx_B]) \, \bbP( \bX_A \in D, \bX_B \in (-\infty, \by_B])
   \\
   & \leq \bbP( \bX_A \in C \vee D, \bX_B \in (-\infty, \bx_B \vee \by_B]) \, \bbP( \bX_A \in C \wedge D, \bX_B \in (-\infty, \bx_B \wedge \by_B]) \notag
\end{align} 
for all \(\bx_B, \by_B \in \bbR^{|B|}\) and all \(C,D \in \mathcal{B}(\bbR^{|A|})\) such that \(C \vee D, C \wedge D \in \mathcal{B}(\bbR^{|A|})\), where 
\begin{align*}
  C \vee D 
  & := \{\bx \vee \by \,:\, \bx \in C, \by \in D \}
  \\
  C \wedge D 
  & := \{\bx \wedge \by \,:\, \bx \in C, \by \in D \}\,.
\end{align*} 
Establishing the equivalence between \eqref{Def:cMTP2AB:Eq2} and \eqref{Def:cMTP2AB:Eq} when \(\bX_A\) has a density \(f_{\bX_A}\) with respect to the product measure \(\Motimes_{i \in A} \bbP^{X_i}\), requires adapting the argument underlying the corresponding equivalence in \cite[Theorem 3.10.14]{muller2002}. This adaptation, together with an analysis of closure properties under suitable modes of convergence, is left for future research. 
}


\section*{Acknowledgments}
This research was funded in whole by the Austrian Science Fund (FWF) [10.55776/P36155] project ReDim: Quantifying Dependence via Dimension Reduction.
The authors gratefully acknowledge the support of the WISS 2025 project 'IDA-lab Salzburg' (20204-WISS/225/197-2019 and 20102-F1901166-KZP).


\appendix

\renewcommand{\thetheorem}{\Alph{section}.\arabic{theorem}}
\renewcommand{\thelemma}{\thetheorem}
\renewcommand{\thecorollary}{\thetheorem}
\renewcommand{\theproposition}{\thetheorem}
\renewcommand{\theexample}{\thetheorem}
\renewcommand{\theremark}{\thetheorem}

\section{Additional results on the relation among positive dependence properties}\label{Sec:App}

Lemma \ref{MTP:Help1} is a variant of \cite[Proposition 3.4]{karlin_1980} and it is used on several occasions throughout the paper.

\begin{lemma} \label{MTP:Help1}
Let \(\mu\) be a product measure on the Borel $\sigma$-algebra of \(\bbR^k\) with $1 \leq k \leq d-1$ and let \(f: \bbR^d \to [0,\infty)\) be a measurable function that is \(\mtp\) and such that \(\bz \mapsto f(\bx,\bz)\) is \(\mu\)-integrable. Then
\begin{enumerate}[(1)]
\item\label{MTP:Help1a} the mapping \((\bx, \by) \mapsto \int_{(-\infty,\by]} f(\bx, \bz) \de \mu(\bz)\) is \(\mtp\).
\item\label{MTP:Help1b} the mapping \((\bx, \by) \mapsto \int_{(\by,\infty)} f(\bx, \bz) \de \mu(\bz)\) is \(\mtp\).
\end{enumerate}
\end{lemma}
\begin{proof}
According to \cite[Proposition 3.4]{karlin_1980}, the mapping 
\begin{align*}
  (\bx, \by) \mapsto \int_{\bbR^k} f(\bx, \bz) \; h(\bz, \by) \de \mu(\bz)\,,
\end{align*}
where \(h(\bz, \by): = \mathds{1}_{(-\infty,\by]}(\bz)\), is \(\mtp\) since \(h\) is \(\mtp\). \eqref{MTP:Help1b} follows by a similar reasoning.
\end{proof}


We shall also need the following result on log-convexity.

\begin{lemma} \label{lem:logconvex}
For a continuous function \(f: [0,\infty) \to (0,\infty)\), the following statements are equivalent:
\begin{enumerate}[{\rm (1)}]
    \item\label{lem:logconvex.1} \(f\) is log-convex.
    \item\label{lem:logconvex.2} The mapping \(x \mapsto \frac{f(x+b)}{f(x+a)}\) is non-decreasing for all \(a,b \in [0,\infty)\) with \(a < b\). 
    \item\label{lem:logconvex.3} The mapping \(x \mapsto \frac{f(x+b)}{f(x)}\) is non-decreasing for all \(b \in (0,\infty)\).
\end{enumerate}
\end{lemma}
\begin{proof}
We first show that \eqref{lem:logconvex.1} and \eqref{lem:logconvex.2} are equivalent.
Note that, for \(a,b \in [0,\infty)\) with \(a < b\), condition \eqref{lem:logconvex.2} is equivalent to the mapping 
\begin{align*}
  x 
  & \mapsto \log \left( \frac{f(x+b)}{f(x+a)} \right) / (b-a)
        =   \frac{(\log \circ f)(x+b) - (\log \circ f)(x+a)}{b-a}
\end{align*} 
being non-decreasing. 
So it remains to show that \(g := \log \circ f\) is convex if and only if 
\begin{align*}
  x 
  & \mapsto \frac{g(x+h) - g(x)}{h}
\end{align*} 
is non-decreasing for every \(h > 0\). But the latter equivalence follows from the fact that \(g\) is convex if and only if 
\begin{align*}
  g(x) - g(x-h)
  & \leq g(x+h) - g(x)
\end{align*} 
for all \(x \in [0,\infty)\) and all \(0 < h \leq x\); cf. 
\cite[Theorem 1.3.1]{niculescu2004} or \cite[Lemma 6]{schmidt2017}.
Finally, clearly \eqref{lem:logconvex.2} and \eqref{lem:logconvex.3} are equivalent.
\end{proof}

We conclude with a modest extension of Example \ref{Ex.Binary} to trivariate binary distributions.

\begin{example}[Binary distributions, \(d=3\)] \label{Ex.Binary.d=3} 
Suppose \(\bX = (X_1, X_2,X_3)\) is a \(3\)-dimensional binary random vector taking on values in \(\{0,1\}^{3}\) with \(\bbP(X_i = 0) = 1 - \bbP(X_i = 1) \in (0,1)\), \(i \in \{1,2,3\}\). Then 
\begin{align*}
\begin{array}{ccccccc}
&&&&& \Nearrow & \LOMTP(\bX) 
\\
\dMTP(\bX) & 
\Longleftrightarrow & \cLOsMTP(\bX) & 
\Longleftrightarrow & \underset{|A|=1}{\cLOMTP(\bX_B|\bX_A)} & 
\\
&&&&& \Searrow & \underset{|A|=1}{\LOSD(\bX_B|\bX_A)}
\end{array}  
\end{align*}
with the reverse implications being invalid.
We verify the result in several steps.\pagebreak
\begin{enumerate}[{\rm (1)}]
\item\label{Ex.Binary:1} 
According to \cite[Section 4.2]{fallat_total_2017}, \(\dMTP(\bX)\) holds if and only if the following nine inequalities hold:
\vspace{-5mm}
\begin{subequations} %
\begin{center}
\begin{minipage}{0.3\linewidth}
\begin{equation} \label{ex:bina_dMTP_a}
p_{011} p_{000} \geq p_{010} p_{001}
\end{equation}
\end{minipage}\hfill
\begin{minipage}{0.3\linewidth}
\begin{equation}\label{ex:bina_dMTP_b}
p_{101} p_{000} \geq p_{100} p_{001}
\end{equation}
\end{minipage}\hfill
\begin{minipage}{0.3\linewidth}
\begin{equation}\label{ex:bina_dMTP_c}
p_{110} p_{000} \geq p_{100} p_{010}
\end{equation}
\end{minipage}

\vspace{0mm} %

\begin{minipage}{0.3\linewidth}
\begin{equation}\label{ex:bina_dMTP_d}
p_{111} p_{100} \geq p_{110} p_{101}
\end{equation}
\end{minipage}\hfill
\begin{minipage}{0.3\linewidth}
\begin{equation}\label{ex:bina_dMTP_e}
p_{111} p_{010} \geq p_{110} p_{011}
\end{equation}
\end{minipage}\hfill
\begin{minipage}{0.3\linewidth}
\begin{equation}\label{ex:bina_dMTP_f}
p_{111} p_{001} \geq p_{101} p_{011}
\end{equation}
\end{minipage}

\vspace{0mm} 

\begin{minipage}{0.3\linewidth}
\begin{equation}\label{ex:bina_dMTP_g}
p_{111} p_{000} \geq p_{100} p_{011}
\end{equation}
\end{minipage}\hfill
\begin{minipage}{0.3\linewidth}
\begin{equation}\label{ex:bina_dMTP_h}
p_{111} p_{000} \geq p_{010} p_{101}
\end{equation}
\end{minipage}\hfill
\begin{minipage}{0.3\linewidth}
\begin{equation}\label{ex:bina_dMTP_i}
p_{111} p_{000} \geq p_{001} p_{110}
\end{equation}
\end{minipage}
\end{center}
\end{subequations}

\item\label{Ex.Binary:2} 
We now study \(\cLOMTP(\bX_B|\bX_A)\) with \(|A|=1\). 
To verify the stated equivalences, we begin by illustrating the main idea in the case \(A = \{1\}\) and \(B = \{2,3\}\).
Therefore, define 
\(\hat{p}_{ijk} := \bbP(X_2 \leq j, X_3 \leq k \mid X_1=i)\), \(i,j,k \in \{0,1\}\).
Then we have 
\begin{align*}
\begin{array}{llll}
\hat{p}_{000} = \frac{p_{000}}{p_{000} + p_{010} + p_{001} + p_{011}} & 
\hat{p}_{010} = \frac{p_{000} + p_{010}}{p_{000} + p_{010} + p_{001} + p_{011}} &
\hat{p}_{001} = \frac{p_{000}+p_{001}}{p_{000} + p_{010} + p_{001} + p_{011}} & 
\hat{p}_{011} = 
1
\\ \\
\hat{p}_{100} = \frac{p_{100}}{p_{100} + p_{110} + p_{101} + p_{111}} &
\hat{p}_{110} = \frac{p_{100} + p_{110}}{p_{100} + p_{110} + p_{101} + p_{111}} &
\hat{p}_{101} = \frac{p_{100} + p_{101}}{p_{100} + p_{110} + p_{101} + p_{111}} & 
\hat{p}_{111} = 
1\,,
\end{array}
\end{align*}
and \(\cLOMTP(\bX_B|\bX_A)\) if and only if the following nine inequalities hold:
\vspace{-5mm}
\begin{subequations} 
\begin{center}
\begin{minipage}{0.3\linewidth}
\begin{equation} \label{ex:bina_cMTP_a}
\hat{p}_{011} \hat{p}_{000} \geq \hat{p}_{010} \hat{p}_{001}
\end{equation}
\end{minipage}\hfill
\begin{minipage}{0.3\linewidth}
\begin{equation}\label{ex:bina_cMTP_b}
\hat{p}_{101} \hat{p}_{000} \geq \hat{p}_{100} \hat{p}_{001}
\end{equation}
\end{minipage}\hfill
\begin{minipage}{0.3\linewidth}
\begin{equation}\label{ex:bina_cMTP_c}
\hat{p}_{110} \hat{p}_{000} \geq \hat{p}_{100} \hat{p}_{010}
\end{equation}
\end{minipage}

\vspace{0mm} 

\begin{minipage}{0.3\linewidth}
\begin{equation}\label{ex:bina_cMTP_d}
\hat{p}_{111} \hat{p}_{100} \geq \hat{p}_{110} \hat{p}_{101}
\end{equation}
\end{minipage}\hfill
\begin{minipage}{0.3\linewidth}
\begin{equation}\label{ex:bina_cMTP_e}
\hat{p}_{111} \hat{p}_{010} \geq \hat{p}_{110} \hat{p}_{011}
\end{equation}
\end{minipage}\hfill
\begin{minipage}{0.3\linewidth}
\begin{equation}\label{ex:bina_cMTP_f}
\hat{p}_{111} \hat{p}_{001} \geq \hat{p}_{101} \hat{p}_{011}
\end{equation}
\end{minipage}

\vspace{0mm} 

\begin{minipage}{0.3\linewidth}
\begin{equation}\label{ex:bina_cMTP_g}
\hat{p}_{111} \hat{p}_{000} \geq \hat{p}_{100} \hat{p}_{011}
\end{equation}
\end{minipage}\hfill
\begin{minipage}{0.3\linewidth}
\begin{equation}\label{ex:bina_cMTP_h}
\hat{p}_{111} \hat{p}_{000} \geq \hat{p}_{010} \hat{p}_{101}
\end{equation}
\end{minipage}\hfill
\begin{minipage}{0.3\linewidth}
\begin{equation}\label{ex:bina_cMTP_i}
\hat{p}_{111} \hat{p}_{000} \geq \hat{p}_{001} \hat{p}_{110}
\end{equation}
\end{minipage}
\end{center}
\end{subequations}

\smallskip
A direct calculation gives
\begin{align*}
\begin{array}{rrr}
\eqref{ex:bina_dMTP_a} \Longleftrightarrow \eqref{ex:bina_cMTP_a} &
\eqref{ex:bina_dMTP_b} \Longleftrightarrow \eqref{ex:bina_cMTP_b} & 
\eqref{ex:bina_dMTP_c} \Longleftrightarrow \eqref{ex:bina_cMTP_c}
\\ \\
\eqref{ex:bina_dMTP_d} \Longleftrightarrow \eqref{ex:bina_cMTP_d} & &
\\ \\
\eqref{ex:bina_dMTP_b}, \eqref{ex:bina_dMTP_c}, \eqref{ex:bina_dMTP_g} \Longrightarrow \eqref{ex:bina_cMTP_g} &
\eqref{ex:bina_dMTP_c}, \eqref{ex:bina_dMTP_h} \Longrightarrow \eqref{ex:bina_cMTP_h} &
\eqref{ex:bina_dMTP_b}, \eqref{ex:bina_dMTP_i} \Longrightarrow \eqref{ex:bina_cMTP_i}
\end{array}
\end{align*}
Moreover, \eqref{ex:bina_cMTP_e} is equivalent to 
\begin{equation} \label{eq:bina}
\bbP(X_1 = 1, X_3 = 1) \, \bbP(X_1 = 0, X_3 = 0) \geq 
\bbP(X_1 = 1, X_3 = 0) \, \bbP(X_1 = 0, X_3 = 1)\,,
\end{equation}
and \eqref{ex:bina_cMTP_f} is equivalent to 
\begin{equation} \label{eq:binb}
\bbP(X_1 = 1, X_2 = 1) \, \bbP(X_1 = 0, X_2 = 0) \geq 
\bbP(X_1 = 1, X_2 = 0) \, \bbP(X_1 = 0, X_2 = 1)\,.
\end{equation}
Since \(\dMTP(\bX)\) is closed under marginalization (Remark \ref{ClosedMargins.2}), it follows that \eqref{eq:bina} and \eqref{eq:binb} hold.
Consequently, \(\dMTP(\bX)\) implies \(\cLOMTP(\bX_B|\bX_A)\) with \(A = \{1\}\) and \(B = \{2,3\}\). This confirms Theorem \ref{Thm.MTP2}\eqref{Thm.MTP2:3}.

\item \label{Ex.Binary:3} 
We now verify that \(\cLOMTP(\bX_B|\bX_A)\) with \(|A|=1\) implies \(\dMTP(\bX)\) using the notation introduced in part \eqref{Ex.Binary:2}. 
The stated equivalences then follow from Theorems \ref{cMTP2:Seq} and \ref{Thm.MTP2}.
\\
First, assume that \(p_{001}, p_{010}, p_{100} > 0\), and note that in this case  \(\cLOMTP(\bX_B|\bX_A)\) with \(|A|=1\) if and only if the inequalities \eqref{ex:bina_cMTP_a}-\eqref{ex:bina_cMTP_i} are satisfied, together with 18 additional inequalities: nine corresponding to \(A = \{2\}\) and \(B = \{1,3\}\), and nine corresponding to \(A = \{3\}\) and \(B = \{1,2\}\).
Similarly to the argument in \eqref{Ex.Binary:2} this implies \eqref{ex:bina_dMTP_a}-\eqref{ex:bina_dMTP_f}, and \eqref{ex:bina_dMTP_g}-\eqref{ex:bina_dMTP_i} follow from \eqref{ex:bina_dMTP_a}-\eqref{ex:bina_dMTP_f} using \(p_{001}, p_{010}, p_{100} > 0\); cf.~\cite[Proposition 2.1]{karlin_1980}.
\\
Now, assume that \(p_{001}=0\) and \(p_{010}, p_{100} > 0\). 
Then \eqref{ex:bina_dMTP_a}, \eqref{ex:bina_dMTP_b} and \eqref{ex:bina_dMTP_i} trivially hold and it remains to show that  \(\cLOMTP(\bX_B|\bX_A)\) with \(|A|=1\) implies \eqref{ex:bina_dMTP_c} - \eqref{ex:bina_dMTP_h}.
Like in \eqref{Ex.Binary:2}, \eqref{ex:bina_dMTP_c} - \eqref{ex:bina_dMTP_f} follow directly from the  \(\cLOMTP(\bX_B|\bX_A)\) inequalities, while \eqref{ex:bina_dMTP_g} and \eqref{ex:bina_dMTP_h} are due to \eqref{ex:bina_dMTP_c} - \eqref{ex:bina_dMTP_e} using \(p_{010}, p_{100} > 0\); cf.~\cite[Proposition 2.1]{karlin_1980}.
The other two cases with only one of \(p_{010}\) or \(p_{100}\) equal to zero are analogous.
\\
Further, assume that \(p_{001} = 0 = p_{010}\) and \(p_{100} > 0\). 
Then \eqref{ex:bina_dMTP_a} - \eqref{ex:bina_dMTP_c} and \eqref{ex:bina_dMTP_h} - \eqref{ex:bina_dMTP_i} trivially hold and it remains to show that  \(\cLOMTP(\bX_B|\bX_A)\) with \(|A|=1\) implies \eqref{ex:bina_dMTP_d} - \eqref{ex:bina_dMTP_g}.
Like in \eqref{Ex.Binary:2}, \eqref{ex:bina_dMTP_d} - \eqref{ex:bina_dMTP_f} follow directly from the  \(\cLOMTP(\bX_B|\bX_A)\) inequalities, while \eqref{ex:bina_dMTP_g} is due to 
\eqref{ex:bina_cMTP_g} using a distinction of cases.
The other two settings with only one of \(p_{010}\) or \(p_{100}\) strictly positive are analogous.
\\
Finally, assume that \(p_{001} = p_{010} = p_{100} = 0\).
Then \eqref{ex:bina_dMTP_a} - \eqref{ex:bina_dMTP_c} and \eqref{ex:bina_dMTP_g} - \eqref{ex:bina_dMTP_i} trivially hold and it remains to show that  \(\cLOMTP(\bX_B|\bX_A)\) with \(|A|=1\) implies \eqref{ex:bina_dMTP_d} - \eqref{ex:bina_dMTP_f} which is immediate.

\item 
That the reverse implication involving \(\MTP(\bX)\) is invalid follows from Example \ref{Ex.Binary}.

\item 
Finally, we present a counterexample that is \(\LOSD(\bX_B|\bX_A)\) with \(|A|=1\) but not \(\cLOMTP(\bX_B|\bX_A)\) with \(|A|=1\): 
Suppose \(\bX \in \{0,1\}^3\) with probabilities of occurrence
\begin{center}
$\displaystyle \begin{array}{ c|c c c c c c c c }
\bx = (x_{1},x_{2},x_{3}) & ( 0,0,0) & ( 1,0,0) & ( 0,1,0) & ( 0,0,1) & ( 1,1,0) & ( 1,0,1) & ( 0,1,1) & ( 1,1,1)\\
\hline
\bbP(\bX = \bx) & 0.35 & 0.15 & 0.15 & 0.15 & 0.05 & 0.05 & 0.05 & 0.05
\end{array}$
\end{center}
Then it is straightforward to verify \(\LOSD(\bX_B|\bX_A)\), but \eqref{ex:bina_cMTP_a} and hence \(\cLOMTP(\bX_B|\bX_A)\) fails. \\

\end{enumerate}
\end{example}
\end{document}